\documentclass{article}

\usepackage{amssymb,amsmath,latexsym,amscd,amsfonts, bbm, bm}  
\usepackage{graphicx}  
\usepackage{url}

\usepackage[normalem]{ulem}
\usepackage[svgnames]{xcolor}

\usepackage{authblk}

\usepackage{newpxtext,newpxmath} 

\usepackage{hyperref}
\hypersetup{
     colorlinks=true,           
     linkcolor=Salmon,          
     citecolor=blue,            
     filecolor=blue,            
     urlcolor=cyan,             
 }

\graphicspath{{Figs}}

\newtheorem{theorem}{Theorem}[section]  
\newtheorem{remark}{Remark}[theorem]  
  
\newtheorem{definition}[theorem]{Definition}  
  
\newtheorem{claim}[theorem]{Claim}  
\newtheorem{lemma}[theorem]{Lemma}

\newtheorem{corollary}[theorem]{Corollary}

\newtheorem{problem}[theorem]{Problem}  
 
\newcommand{\Eq}[1]{Eq.~(\ref{#1})}
\newcommand{\Fig}[1]{Fig.~\ref{#1}}
\newcommand{\Def}[1]{Definition~\ref{#1}}
\newcommand{\Lem}[1]{Lemma~\ref{#1}}

\newcommand{\Cor}[1]{Corollary~\ref{#1}}
\newcommand{\Sec}[1]{Sec.~\ref{#1}}
\newcommand{\cRef}[1]{Ref.~\cite{#1}}

\newcommand{\cc}[1]{~\cite{#1}}

\newcommand{\Thm}[1]{Theorem~\ref{#1}}
\newcommand{\App}[1]{Appendix~\ref{#1}}

\newcommand{\qedsymb}{\hfill{\rule{2mm}{2mm}}}  
\newenvironment{proof}[1][]{\begin{trivlist}  
\item[\hspace{\labelsep}{\bf\noindent Proof#1:\/}] 
 }{\qedsymb\end{trivlist}}

\newcommand{\ignore}[1]{}

\newcommand{\norm}[1]{{\| #1 \|}}  
  
\newcommand{\ket}[1]{{ |{#1} \rangle }}  
\newcommand{\bra}[1]{{ \langle {#1} | }}
\newcommand{\av}[1]{{ \langle {#1} \rangle }}
\newcommand{\braket}[2]{{ \langle {#1} | {#2} \rangle}}
\newcommand{\ketbra}[2]{{ |{#1} \rangle\langle {#2} | }}

\newcommand{\suppress}[1]{}

\newcommand{\EqDef}{\stackrel{\mathrm{def}}{=}}
\newcommand{\Id}{\mathbbm{1}}

\DeclareMathOperator{\Prob}{Prob}
\DeclareMathOperator*{\Tr}{Tr}
\DeclareMathOperator{\poly}{poly}
\DeclareMathOperator*{\minarg}{minarg}

\newcommand{\ualpha}{{\underline{\alpha}}}
\newcommand{\ubeta}{{\underline{\beta}}}

\newcommand{\mcM}{{\mathcal{M}}}

\newcommand{\BBR}{{\mathbb{R}}}
\newcommand{\BBC}{{\mathbb{C}}}
\newcommand{\BBH}{\mathbbm{H}}

\newcommand{\Sqq}{{\mathbb{S}^{(k)}_{qq}}}
\newcommand{\SQ}{{\mathbb{S}_q}}
\newcommand{\eps}{\epsilon}
\newcommand{\gs}{\Omega}

\newcommand{\Xlow}{X_\mathrm{low}}

\DeclareMathOperator*{\supp}{supp}

\newcommand{\sym}{{\mathrm{sym}}}
\newcommand{\LH}{\mathsf{LH}}
\newcommand{\qqLH}{\mathsf{qqLH}}

\begin{document}

\title{Quasi-quantum states and the quasi-quantum PCP theorem}
\author[1]{Itai Arad\footnote{{Email:
 \texttt{arad.itai@fastmail.com}}}}
\author[1,2]{Miklos Santha}
\affil[1]{Centre for Quantum Technologies, Singapore}
\affil[2]{CNRS, IRIF, Universit\'e Paris Cit\'e, France}
\date{\today}
\maketitle

\begin{abstract}
  
  We introduce $k$-local quasi-quantum states: a superset of the
  regular quantum states, defined by relaxing the positivity
  constraint. We show that a $k$-local quasi-quantum state on $n$
  qubits can be 1-1 mapped to a \emph{distribution of assignments}
  over $n$ variables with an alphabet of size $4$, which is subject
  to non-linear constraints over its $k$-local marginals. Therefore,
  solving the $k$-local Hamiltonian over the quasi-quantum states is
  equivalent to optimizing a distribution of assignments over a
  classical $k$-local CSP.  We show that this optimization problem
  is essentially classical by proving it is NP-complete. Crucially,
  just as ordinary quantum states, these distributions lack a simple
  tensor-product structure and are therefore not determined
  straightforwardly by their local marginals. Consequently, our
  classical optimization problem shares some unique aspects of
  Hamiltonian complexity: it lacks an easy search-to-decision
  reduction, and it is not clear that its 1D version can be solved
  with dynamical programming (i.e., it could remain NP-hard).
  
  Our main result is a PCP theorem for the $k$-local Hamiltonian
  over the quasi-quantum states in the form of a
  hardness-of-approximation result. The proof suggests the existence
  of a subtle promise-gap amplification procedure in a model that
  shares many similarities with the quantum local Hamiltonian
  problem, thereby providing insights on the quantum PCP conjecture.
  
\end{abstract}

\section{Introduction}

\subsection{Quantum vs. classical constraint satisfaction problems}

The fields of classical constraint satisfaction problems (CSPs) and
Quantum Hamiltonian complexity share many similarities. In both
fields, a central problem is the minimization of a global cost
function that is the sum of local constraints. In the classical
world, this is the well-known MAX-$k$-SAT problem, whereas in
Hamiltonian complexity this is the $k$-local Hamiltonian ($k$-LH)
problem. As a decision problem, MAX-$k$-SAT is known to be
NP-complete\cc{ref:Garey1976-NP}, and similarly in the quantum world
the $k$-LH problem is QMA-complete\cc{ref:Kitaev1999-LH,
ref:Kitaev2002-QI}.  Furthermore, in both cases we can use
reductions to identify restricted instances of problems that are
still as hard as the general problem. For example, in both cases the
problem remains hard when $k=2$\cc{ref:Garey1976-NP,
ref:Oliveira2008-2D}. On the other hand, in both problems, one can
exploit the restricted structure of the $k=2$ case when we are asked
to decide whether all constraints can be satisfied or not and show
that in such case the the problem becomes easy. This is the 2-SAT
problem in the classical case, and the 2-local frustration-free
problem in the quantum case, both of which are in
P\cc{ref:Davis1962-2SAT, ref:Bravyi2006-2QSAT}.  Moreover, in both
cases this problem becomes hard for $k=3$: the classical $3$-SAT
problem is NP-hard \cc{ref:Karp72-NP}, the frustration-free $k=3$
local Hamiltonian problem is QMA-hard\footnote{To be more precise,
in such case the problem is known to be $\mathrm{QMA}_1$-hard, which
is a variant of QMA with 1-sided error --- see
\cc{ref:Bravyi2006-2QSAT} for details.}\cc{ref:Gosset2013-3QSAT}.

Nevertheless, there are some striking differences between the
quantum and the classical cases. Below, we list some of the central
ones.
\begin{description}
  \item [1. Circuit-to-CSP reduction:] 
    In the reduction that shows that the 3-SAT problem is NP-hard,
    one follows $T$ computational steps of a classical circuit over
    $n$ bits and turns them into a local CSP.  Every time-step $t$
    is identified with a fresh set of $n$ bits, and the CSP is made
    of local consistency conditions between the bits of step $t$ to
    the bits of step $t+1$. The entire CSP is therefore defined over
    $O(nT)$ bits.
    
    In the quantum case, it is not clear how to check \emph{local}
    consistency between two consecutive time steps $\ket{\psi_t}$
    and $\ket{\psi_{t+1}}$ due to the entanglement within each time
    step. Inspired by the pioneering work of
    Feynman\cc{ref:Feynman1986-QC}, Kitaev overcame this problem
    using the `Hamiltonian Clock' construction\cc{ref:Kitaev1999-LH,
    ref:Kitaev2002-QI} (but see also
    \cRef{ref:Anshu2023-Cir-to-H} for an alternative reduction).
    There, the entire history of the quantum circuit is stored as an
    entangled history state $\ket{\phi} =
    \frac{1}{\sqrt{T}}\sum_{t=1}^T \ket{\psi_t}\otimes\ket{t}$ over
    a \emph{single} set of $n$ qubits, together with $O(1)$ clock
    qubits. All together, the resultant local Hamiltonian contains
    $n + O(1)$ qubits.
    
  \item [2. CSPs in 1D:]
    Given a $k$-local CSP over bits that are arranged on a line, if
    the constraints involve only geometrically contiguous sets of
    bits, e.g., $i, i+1, \ldots, i+k-1$, then both the $k$-SAT and
    the MAX-$k$-SAT problems can be solved efficiently using
    dynamical programming. On the other hand, the equivalent quantum
    problem of solving the 1D $k$-local Hamiltonian problem is known
    to remain QMA-hard\cc{ref:Aharonov2009-1D, ref:Nagaj2008-PhD,
    ref:Hallgren2013-1D}.
    
  \item [3. Search-to-decision reduction:]
    In classical CSP problems, one can naturally reduce the search
    problem of finding an optimal assignment to the decision problem
    of whether or not a satisfying assignment exists. One simply
    builds the satisfying assignment bit-by-bit, by conditioning the
    CSP on the previously discovered bits. In the quantum case,
    there does not seem to be such a natural reduction that reduces
    the promise problem of a $k$-local Hamiltonian (i.e., whether
    the ground energy is smaller than some constant $a$ or is larger
    than some other constant $b$), and turns that into a low-energy
    state. In fact, it was shown in \cc{ref:Irani2021-S-to-C} that
    such reduction does not exist relative to a particular quantum
    oracle (but see \cc{ref:Weggemans2024-S-to-C} for a related
    problem for which such reduction exists).
    
  \item [4. The PCP theorem and the qPCP conjecture:]
    Whereas for the celebrated PCP theorem there exist several
    proofs\cc{ref:Arora1998-PCP1, ref:Arora1998-PCP2,
    ref:Dinur2007-PCP}, the analogous statement for quantum local
    Hamiltonian is only a conjecture\cc{ref:Aharonov2013-qPCP} (but
    see, for example, \cRef{ref:Buherman2024:qPCP} and references
    within for recent results about this conjecture). We do not know
    if deciding whether the ground energy of a local quantum
    Hamiltonian is below $a$ or above $b$, when $b-a$ is
    proportional to the size of the system, is still QMA hard.
\end{description}

Arguably, the source of these differences lies in the fact that
while a classical assignment has a simple tensor-product structure
and is straightforwardly determined by its local marginals, this is
not the case for quantum states. There can be two orthogonal quantum
states that share the same set of marginals (reduced density
matrices). Moreover, given a set of consistent local marginals of a
classical assignment which cover the entire set of bits, it is easy
to construct the full assignment.  We shall refer to this property
as the ``bottom-up'' property. Quantum states do not possess this
property. In the quantum world, it is even QMA-hard to decide if a
set of locally-consistent local-marginals is also consistent with a
global quantum state\cc{ref:Liu2006-RDMs, ref:Broadbent2022-RDMs}.
In that sense, a quantum state might be more like a
\emph{distribution} of classical assignments, which also lacks the
``bottom-up'' property.

Following the above discussion, in this work we introduce an
optimization problem of local constraints over a convex subset of
distributions of assignments. We show that this is an NP-complete
problem, which can therefore be viewed as a classical problem.
{Nonetheless and importantly}, the distributions we optimize over
lack the bottom-up property, which leads to many similarities with
the quantum $k$-local Hamiltonian problem. For example, there does
not seem to be a straightforward way to use dynamical programming to
solve its 1D version, and {the problem} lacks a simple
search-to-decision reduction. Nevertheless, we do manage to prove a
PCP theorem as a hardness-of-approximation statement. We show that
our optimization problem remains NP-hard also when asked to
approximate the minimal expected number of violations up to a
constant factor of the system size. Our proof may therefore give
valuable insights in the quantum PCP conjecture and the possibility
of promise gap amplification in local Hamiltonians.

\subsection{Quasi-quantum states}

While our problem can be viewed as an optimization problem over
distributions of classical assignments, it is formally defined as
the usual $k$-local Hamiltonian problem solved over a
\emph{superset} of the (mixed) quantum states. We call these
operators \emph{$k$-local quasi-quantum states (qq states)}, where
$k$ is some fixed locality parameter that will always be taken to be
identical to the locality of the Hamiltonian at hand.
Mathematically, $k$-local qq states are a convex set of multiqubit
operators, obtained by relaxing the positivity constraint of the
usual quantum states. The exact definition of these operators in
given in \Sec{sec:qq-states}. Here we provide an informal
definition.

Recall the definition of a general quantum state over $n$ qubits. It
is a Hermitian operator $\rho$ with unit trace and a global
positivity constraint $\rho\succeq 0$. The $k$-local qq states are
defined by relaxing the positivity constraint.  The definition uses
a fixed single-qubit symmetric, informationally complete POVM
(SIC-POVM) measurement, which we denote by $\{F_0, F_1, F_2, F_3\}$
(see \Sec{sec:SIC-POVM} for full discussion). Given a single-qubit
quantum state $\rho$, we can measure it using the SIC-POVM
measurement and obtain result $\alpha\in \{0,1,2,3\}$ with
probability $\Prob(\alpha) = \Tr(F_\alpha \rho)$. This establishes a
linear map from the space of single-qubit quantum states to a
probability distribution over an alphabet of size $4$. As our
measurement is informationally-complete, the map $\rho \mapsto
\Prob(\alpha)$ is invertible; all the information of the operator
$\rho$ can be recovered from the probability distribution
$\Prob(\alpha)$. 

Taking tensor products of the SIC-POVM, we can naturally extend our
SIC-POVM to a multiqubit setup, resulting in a SIC-POVM measurement
with $4^n$ elements of the form $F_\ualpha \EqDef
F_{\alpha_1}\otimes\ldots \otimes F_{\alpha_n}$, where
$\ualpha\EqDef (\alpha_1, \ldots, \alpha_n)$. As in the single-qubit
case, the measurement probability $\Prob(\ualpha) \EqDef \Tr(\rho
F_\ualpha)$ establishes an invertible map between the space of
$n$-qubits states and multivariate distributions $\mu(\ualpha)$ over
the alphabet $\{0,1,2,3\}^n$.
 
With a multiqubit SIC-POVM in our hands, we are now ready to define
the $k$-local qq states. We say that an $n$-qubits operator $X$ is a
$k$-local qq state if $\Tr(X)=1$, $X=X^\dagger$ and in addition it
satisfies the following positivity requirements: 
\begin{description}
  \item[1. Global positivity:] $\Tr(X F_\ualpha)\ge 0$ for 
    every SIC-POVM element $F_\ualpha =
    F_{\alpha_1}\otimes\ldots\otimes F_{\alpha_n}$
    
  \item[2. Local positivity:] For every subset of $k$ 
    qubits $I=(i_1, \ldots, i_k)$, the reduced operator $X^{(I)}
    \EqDef \Tr_{\setminus I}X$ (i.e., the operator obtained by
    tracing out all qubits except for those in $I$) satisfies
    $X^{(I)}\succeq 0$.
\end{description}
A rough physical interpretation of this definition is that we can
measure a qq state \emph{globally} using our fixed SIC-POVM
measurement, and that \emph{locally}, our state looks just like a
regular quantum state (i.e., it is non-negative). Mathematically,
the global positivity requirement means that there is a 1-to-1
mapping between $X$ and a probability distribution
$\mu(\ualpha)=\Tr(XF_\ualpha)$. The local positivity condition
restricts the subset of distributions; only distributions
$\mu(\ualpha)$ that give rise to operators $X$ such that
$X^{(I)}\succeq 0$ for every $k$-local subset $I=(i_1, \ldots, i_k)$
are valid. As we show in \Sec{sec:qq-states}, this condition can be
phrased as a non-linear condition on the \emph{marginal} of
$\mu(\ualpha)$ on the subset $I$.

With some abuse of notation, we can pictorially view the $k$-local
qq states as a convex superset of the set of quantum states, and as
a subset of the convex polytope of all possible probability
distributions over the alphabet $\{0,1,2,3\}^n$, as shown in
\Fig{fig:qq-states}.
\begin{figure}
\centering 
  \includegraphics[scale=0.2]{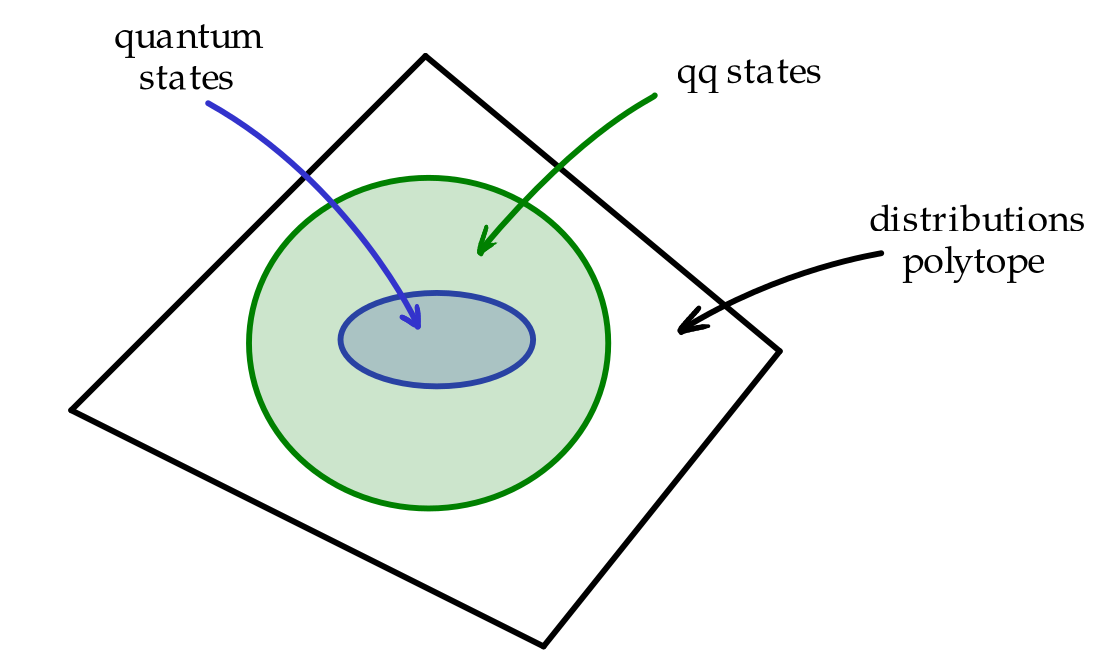}  
  \caption{A schematic representation of the relations
    between quantum states, $k$-local qq states and general
    probability distributions. The $k$-local qq states are a
    superset of the quantum states, and are a convex subset of the
    polytope of all possible probability distributions of
      assignments over the alphabet $\{0,1,2,3\}^n$. }
  \label{fig:qq-states}
\end{figure}

The above picture motivates us to define our optimization problem as
the usual $k$-local Hamiltonian problem, but this time optimized
over the set of $k$-local qq states. For every subset of qubits
$I=(i_1, \ldots, i_k)$, the reduced operator $X^{(I)}=\Tr_{\setminus
I}X$ is determined by the marginal $\mu^{(I)}(\alpha_{i_1}, \ldots,
\alpha_{i_k})$ of the corresponding distribution. It follows that
our optimization problem is equivalent to minimizing the
\emph{expected} violation number of a set of $k$-local constraints
over the subset of allowed distributions. By construction, the
minimum is a lower bound to the actual quantum ground energy, which
is the minimal energy over all quantum states.

Note that the subset of distributions that corresponds to the set of
$k$-local qq states is a convex set that is \emph{strictly smaller}
than the set of all possible distributions. This is crucial for the
lack of the bottom-up property. Indeed, since we are looking for a
distribution that minimizes the \emph{expected} number of
violations, which is a linear function of the distribution, the
optimum will be achieved by an extreme point of the convex set of
distributions. Had we used the full probability polytope, its
extreme points would be pure assignments, i.e., sharp distributions,
and we will be back to the usual CSP problem. Our model is therefore
defined so that on one hand the extreme points of the subset of
distributions are \emph{not} sharp distributions, and on the other
hand, one can verify that a given distribution is indeed within our
subset by looking at its local marginals.

Let us now give an informal statement of the quasi-quantum PCP
theorem. We consider $k$-local Hamiltonians over $n$ qubits of the
from $H=\sum_i h_i$, where every $h_i$ is a Hermitian operator with
$O(\poly(n))$ norm, which is acting non-trivially on at most $k$
qubits. The quasi-quantum ground energy of $H$ is defined as
$\eps_0^{qq} \EqDef \min_{X\in \Sqq} \Tr(XH)$, where $\Sqq$ denotes
the set of all $k$-local qq operators. Clearly $\eps_0^{qq}$ is a
\emph{lower bound} to the actual ground energy of $H$, as the qq
states are superset of the quantum states. 

\begin{theorem}
\label{thm:qqpcp}
  Consider the following problem. We are given a $3$-local
  Hamiltonian $H=\sum_i h_i$ together with two numbers $a<b$ and a
  promise that either $\eps_0^{qq}\le a$ ({\rm YES} case) or
  $\eps_0^{qq}\ge b$ ({\rm NO} case), where $\eps_0^{qq}$ is
  the quasi-quantum ground energy over $3$-local quasi-quantum
  states. We are asked to decide between these cases. Then this
  problem is {\rm NP}-complete for a promise gap $b-a$ going from
  $b-a=e^{-O(n)}L_H$ up to $b-a=\Theta(L_H)$, where $L_H$ is the
  global energy scale $L_H\EqDef \sum_i \norm{h_i}$.
\end{theorem}
The formal statement and proof of the above theorem is found in
Theorems~\ref{thm:inside-NP} and \ref{thm:main-NP-hard}.

\subsection{Overview of the proof}

To show that our problem is inside NP, we follow a construction by
I.~Pitowsky\cc{ref:Pitowsky1991-CorrPolytope}, which relies on the
Carath{\'e}odory's theorem\cc{ref:Caratheodory1911} to show that for
every multivariate distribution $\mu(\ualpha)$ there exists a
``sister distribution'' $\mu'(\ualpha)$ such that:
\begin{enumerate}
  \item $\mu$ and $\mu'$ share the same $k$-local marginals
  \item The support of $\mu'$ contains at most $O(n^k)$ assignments. 
\end{enumerate}
Since for every $k$-local set of qubits $I=(i_1, \ldots, i_k)$ the
reduced operators $X^{(I)}$ are determined by the corresponding
$k$-local marginal $\mu^{(I)}(\alpha_{i_1}, \ldots, \alpha_{i_k})$
of the associated probability distribution $\mu$, it follows that
for every $X$ there is a sister qq state $X'$ with a polynomial
support that shares the same $k$-local reduced operators. The prover
will send the description of the sister state to the verifier as a
witness, and the verifier will use it to verify that: (i) it is a
legal qq state, by verifying that the resultant $k$-local reduced
operators are PSD, and (ii) that the resultant energy (which only
depends on the $k$-local reduced operators) is lower than $a$. All
of this can be done to an exponential accuracy given a
$\poly(n)$-bits description of the sister state.

Going to the hardness part, the proof becomes much more involved,
even if we only want to prove hardness for a promise gap of $b-a =
L_H/\poly(n)$.  There are two natural approaches to show
NP-hardness, both of which start from an NP-hard local CSP. In the
first approach, we map the CSP to a classical Hamiltonian that is
diagonal in the computational basis, and in the second approach we
map to a Hamiltonian that is written in terms of the POVM basis
$\{F_\alpha\}$. Let us first consider the first approach. Here one
can take a 3-CNF and turn local clauses into 3-local Hamiltonian
terms in the computational basis. For example, $\ldots (x_1\vee
x_2\vee \neg x_3)\wedge(x_2\vee \neg x_3\vee\neg x_4)\ldots $ will
be mapped into into the $3$-local Hamiltonian $H = \ldots +
\ketbra{001}{001}_{123} + \ketbra{100}{100}_{234} + \ldots$. The
resulting $H$ is diagonal in the computational basis, and therefore
its \emph{quantum} ground energy is its minimal eigenvalue, which is
exactly the minimal number of possible violations. 

To show hardness we need to argue that if the 3-CNF formula is
satisfiable, $H$ has a quasi-quantum ground energy $\eps_0^{qq}\le
0$, whereas if it is not then $\eps_0^{qq}\ge 1/\poly(n)$.  In such
case, completeness is evident but soundness is not guaranteed.
Indeed, if there is a satisfying assignment the quantum ground
energy of $H$ is $\eps_0=0$, and as the qq ground energy is a lower
bound to the quantum ground energy, it follows that also
$\eps_0^{qq}\le 0$. However, when the CSP is not satisfiable, we run
into trouble. In such case the quantum ground energy is $\ge 1$, but
the qq ground energy can still be $0$. For example, consider the
system depicted in \Fig{fig:3qubits}, which consists of a $2$-local
CSP over 3 variables $x_1, x_2, x_3$, with 3 clauses $C_{12},
C_{23}, C_{13}$, where each clause $C_{ij}$ is satisfied iff
$x_i\neq x_j$. The computational Hamiltonian in this case is
\begin{align*}
  H_{\rm comp} = \sum_{\av{i,j}} \Pi_{ij}, \qquad 
    \Pi_{ij} = \ketbra{00}{00}_{ij}
    + \ketbra{11}{11}_{ij} .
\end{align*}
Clearly, the CSP is unsatisfiable (and therefore its proper quantum
ground energy $\eps_0\ge 1$), but its qq ground energy is
$\eps^{qq}_0=0$.  To see that, consider the $2$-local qq state
\begin{align*}
  X \EqDef \frac{1}{2}\big( \ketbra{110}{110} + \ketbra{011}{011}
    +\ketbra{101}{101} - \ketbra{111}{111}\big) .
\end{align*}
It is easy to see that $X$ is indeed a qq state: a direct
calculation shows that it satisfies the global positivity condition
with respect to the standard SIC-POVM basis for qubits (see
Definition~\ref{def:SIC-POVM} and \Eq{eq:POVM-d2} in
\Sec{sec:SIC-POVM} for the precise SIC-POVM definition), and it also
satisfies the local positivity condition since all its 2-local RDMs
are $X^{(ij)} = \frac{1}{2}\big(\ketbra{01}{01} +
\ketbra{10}{10}\big) \succeq 0$. Yet for every edge $\{i,j\}$, we
have $\Tr(\Pi_{ij} X^{(ij)})=0$, hence $\eps_0^{qq}\le \Tr(HX) = 0$.
\begin{figure}
  \begin{center}
    \includegraphics[scale=1]{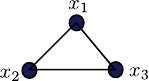}
  \end{center}
  \caption{An illustration of an unsatisfiable classical CSP, where
    every variable $x_i$ takes value in $\{0,1\}$, and on every edge
    we place a constraint that is satisfied only if its adjacent
    variables take different values.} 
\label{fig:3qubits}
\end{figure}

The second natural approach is to write $H$ in terms of the SIC-POVM
basis. Such case makes it easy to write a classical $k$-local CSP
over an alphabet $\{0,1,2,3\}$ as a $k$-local Hamiltonian. A simple
choice is to use the vertex $4$-coloring problem. Given a graph
$G=(V,E)$, we shall define a local Hamiltonian term for every edge
$e=\{i,j\}\in E$ by $h_e \EqDef \sum_{\alpha=0}^3 F_\alpha^{(i)} \otimes
F_\alpha^{(j)}$ and then the total Hamiltonian is $H_G = \sum_{e\in
E} h_e$, which is actually the anti-ferromagnetic Heisenberg
Hamiltonian (see \Sec{sec:4-coloring} for details). In such
case, every $2$-local qq state that is associated with a
distribution $\mu(\ualpha)$, can be viewed as a distribution over
different colorings of $G$. The qq energy $\Tr(H_G X)$ is then the
\emph{expected} number of violations over this distribution. To
prove hardness, we would like to show that if $G$ is $4$-colorable,
then $\eps_0^{qq}=0$, whereas if $G$ is not $4$-colorable, then
$\eps_0^{qq}\ge 1$ (or something similar for general $a<b$ numbers).
The second case is easy to show. If $G$ is not colorable and $X$ is
given by some distribution $\mu(\ualpha)$, then since for each
coloring $\ualpha$ in the support of $\mu$ the number of violations
is at least $1$, it follows that also the expected number of 
violations is $\ge 1$, and so $\Tr(XH_G)\ge 1$. The problem arises
when $G$ is colorable. Assume that in such case it has a valid
coloring $\ualpha^*=(\alpha_1^*, \alpha_2^*, \ldots, \alpha_n^*)$.
Had it been possible to construct a $2$-local qq state $X^*$ whose
support only contains $\ualpha^*$, we would have gotten $\Tr(H_G
X^*)=0$. However, there is no such qq state; by having only one
assignment in its support, the operator $X^*$ will necessarily
violate the local positivity condition. Its local reduced operator
will always contain negative eigenvalues. 

Intuitively, to tackle this problem, one has to choose an NP-hard
problem in which the YES instance has many solutions. Taking a
distribution made of all these assignments will possibly give rise
to an operator $X$ that satisfies local positivity. Our approach is
more general. Instead of asking for many perfect solutions in the
YES case, we construct a problem where in the YES case there are
many low energy assignments (equivalently, low number of
violations), while in the NO case all assignments have a larger
energy. From a many-body physics perspective, this means that in
both YES and NO cases the system is frustrated, only that in the NO
case, it is more frustrated. To construct such systems, we must be
able to have fine control on the amount of frustration in the
system. This is achieved using an idea that we call a `scapegoat'
(see \Sec{sec:bird-eye}). Essentially, we add few more qubits to the
system that we call scapegoat qubits, and we condition the original
terms on these qubits being in a certain `inactive state'. This
changes the $2$-local Hamiltonian to a $3$-local Hamiltonian. The
scapegoat qubits allows us to shift the energy penalties of the many
constraints to a single constraint on the scapegoat qubits, thereby
drastically decreasing the overall energy penalty. 

As we discuss in \Sec{sec:bird-eye}, a single scapegoat qubit is not
enough, and we have to use three such qubits, where each qubit is
conditioning a different part of the system. We must then make sure
that the partial set of the constraints associated with a scapegoat
qubit has many satisfying solutions, from which we can construct
what we call a $\lambda$-solution --- a $2$-local qq state that is
\emph{strictly} positive on its $2$-local marginals. 

To some extent, this approach can be compared to the quantum
circuit-to-Hamiltonian construction\cc{ref:Kitaev2002-QI}, where a
QMA verification circuit is mapped to a Hamiltonian of the form $H =
H_{\rm in} + H_{\rm prop} + H_{\rm out}$. There, each of the $3$
sub-Hamiltonians has a large number of zero energy solutions (i.e.,
each one of them is frustration-free with a large groundspace), but
taking the sum of them introduces some unavoidable frustration due
to incompatibilities between their groundspaces.

Utilizing expander graph techniques from Dinur's PCP
proof\cc{ref:Dinur2007-PCP}, together with standard results on graph
coloring problems, we show that the scapegoat framework is general
enough to be used in reductions from coloring problems that preserve
large promise gaps. We can therefore start from a $4$-coloring
problem with a large promise gap (i.e., a problem that is the result
of the classical PCP theorem), and reduce it to a $3$-local
Hamiltonians with a large promise gap, thereby proving the
quasi-quantum PCP theorem.

\subsection{Discussion}

The $k$-local qq states and their related $k$-local Hamiltonian 
optimization problem form a non-trivial test bed for understanding
the complexity of the $k$-local Hamiltonian problem. Varying the
locality parameter $k$, we gradually move from the classical CSP
problem ($k=0$), which is NP-complete, to solving the problem over
the $k$-local qq states for $k=O(1)$, which is still NP-complete, to
the full quantum Hamiltonian problem at $k=n$, which is
QMA-complete:
\begin{center}
  \includegraphics[scale=0.2]{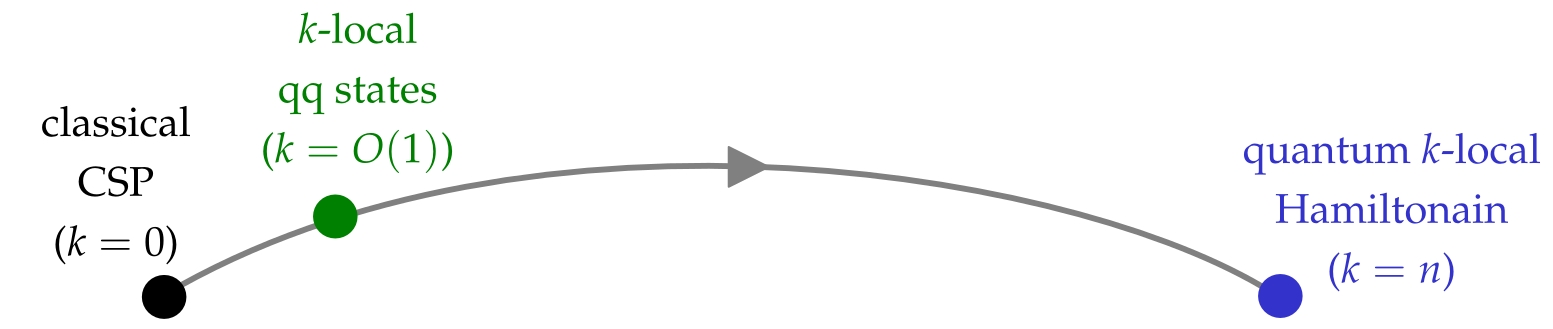}
\end{center}
Our quasi-quantum PCP theorem shows that at least in the first two
steps above, the complexity of the problem remains the same when we
are asked to approximate the ground energy up to an error that
scales with the system size. 

What does the existence of quasi-quantum PCP theorem tells us about
the quantum PCP conjecture? The first point to consider is the
existence of frustration. Frustration measures the minimal expected
number of violations. Our final Hamiltonian has a large degree of
frustration both in the YES and NO cases. This is in stark contrast
with classical case, where the YES solution has no violations, which
translates to perfect completeness. This suggests that if we are
trying to prove the quantum PCP conjecture, we might want to
consider imperfect completeness. In terms of local Hamiltonians,
this would lead to frustration, which we might be able to control
using a mechanism similar to the scapegoat mechanism we used here.

The second, and more important, point is the question of existence
of a gap-amplification procedure for local Hamiltonians.  In Dinur's
PCP proof, the promise gap is iteratively amplified using $O(\log
n)$ amplification cycles. An amplification cycle consists of 3
steps: 1) degree reduction, 2) gap amplification, 3) alphabet
reduction. It takes us back to a local Hamiltonian with the same
locality and local degree we started with, but with a larger promise
gap.

Our PCP proof starts from a gap-amplified classical CSP (a
$4$-coloring problem), which is then reduced to a quasi-quantum
local Hamiltonian problem by a map that preserves the large
promise gap. Therefore, we do not explicitly apply the three steps
of amplification cycle to quasi-quantum states. Just as in the
quantum case, it is not clear that applying such steps on qq states
is even possible. Consider, for example, the degree-reduction step.
In the classical case, one needs to add more variables to the system
and introduce identity constraints $C_{=}(x_i, x_j) = \delta_{x_i,
x_j}$ between them:
\begin{center}
  \includegraphics[scale=0.2]{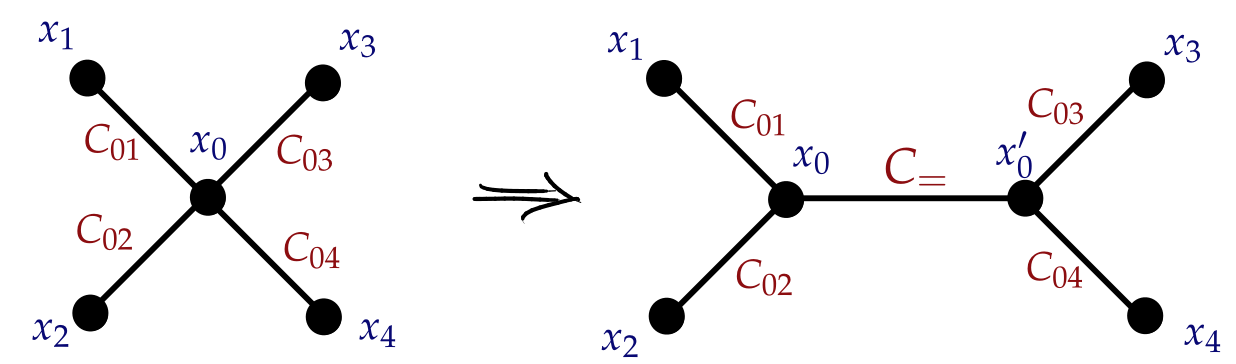}
\end{center}
This can be easily done when the underlying solution is an
assignment. For example, in the CSP above, if $(x_0, x_1, x_2, x_3,
x_4)$ is solution of the original CSP, then $(x_0, x_0, x_1, x_2,
x_3, x_4)$ is the solution of the new CSP. But in the qq states
case, such mapping is impossible. Imposing an identity constraint
between $x_0$ and $x_0'$ implies that the marginal of the
distribution of that corresponds to the qq state of the new CSP 
will have to be supported by the assignments $(0,0), (1,1), (2,2),
(3,3)$. But as we show in \Lem{lem:q-support}, such distribution
violates the local positivity condition; the support of any PSD on 2
qubits must contain at least $9$ assignments. This behavior is
similar to the quantum case, where one cannot impose local equality
constraints between two qubits\cc{ref:Aharonov2013-qPCP}. 

It seems even harder to `quasi-quantize' the gap amplification step.
That step takes a constraint graph system to another constraint
graph system with a much larger alphabet and local degree, while
also amplifying the promise gap. In the quasi-quantum case, as the
locality of the qq states follows the locality of the Hamiltonian,
this would force us to increase the non-locality parameter of the qq
states (i.e., the size of the marginals that must be non-negative).
This changes the underlying configuration space on which the problem
is defined, and seems unlikely such mapping exists.

In the quantum case, the question of existence of the
gap-amplification step is intriguing. As shown by Brand\~ao and
Harrow\cc{ref:Brandao2013-prod-states}, if such procedure existed
for $2$-body quantum Hamiltonians, it would actually
\emph{disprove} of the quantum PCP conjecture: it would imply the
existence of a low-energy product state for the amplified local
Hamiltonian, which could serve as a classical witness, thereby
placing the problem inside NP. This `Catch-22'-like conflict is
arguably our strongest evidence against the correctness of the
quantum PCP conjecture.  

The quasi-quantum PCP theorem proof suggests two possible detours
around this obstacle. First, the Brand\~ao and Harrow\ no-go theorem
assumes a $2$-body Hamiltonian (i.e., local Hamiltonians in which
each local term involves only $2$ qudits). As explicitly shown in
\cRef{ref:Anshu:no-nogo} using an elegant argument, the theorem does
not hold for $4$-body Hamiltonians, and therefore it might also not
hold for our Hamiltonian, which is a $3$-body Hamiltonian. More
generally, the obstacle might be bypassed as illustrated in the
following figure:
\begin{center}
  \includegraphics[scale=0.2]{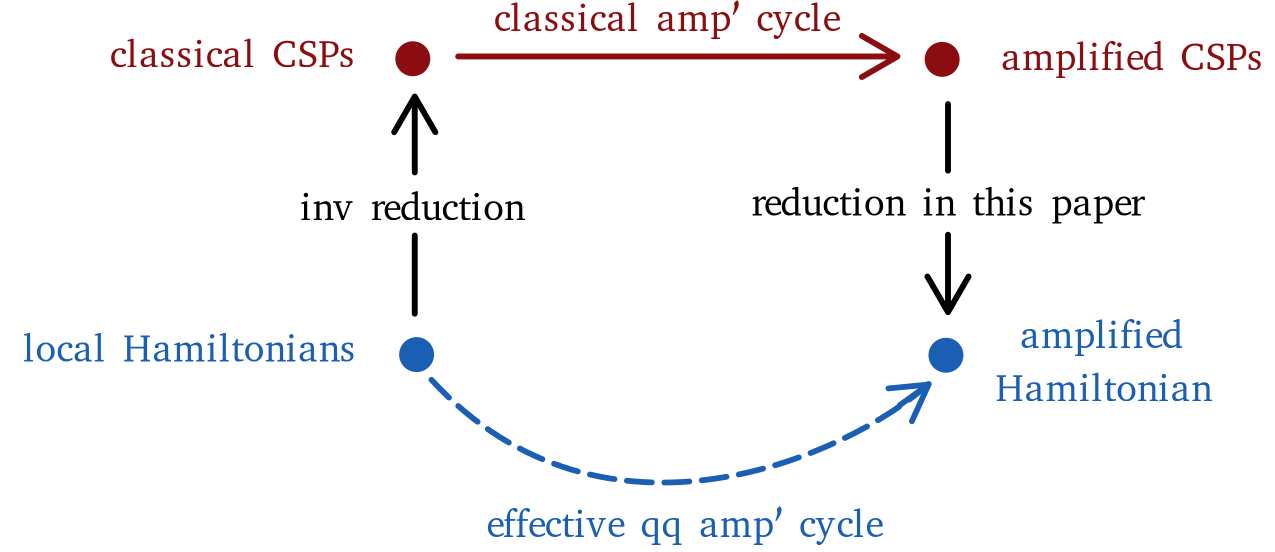}
\end{center}
Starting from a 3-local Hamiltonian, we can apply the inverse of the
reduction we defined in the paper to take us to a 4-coloring
problem. We then apply the full amplification cycle of Dinur, and
then apply our reduction, which takes us back to a gap-amplified
3-local Hamiltonian. This way, we effectively obtain an
amplification cycle for $3$-local Hamiltonians in the quasi-quantum
setup. This observation raises the possibility that something
similar might be possible also in the quantum case. While it might
be impossible to directly quantize the \emph{single} steps in
Dinur's amplification cycle, it may still be possible to quantize
the \emph{full} cycle. In particular, such quantization might
involve a procedure that takes us out of the proper quantum states
space to different objects (for example, operators that are not
necessarily positive), and then back to proper quantum states. 

{~}

We conclude this section by remarking that the qq-states framework
might also be useful as a practical optimization framework for
quantum Hamiltonians. Given a $k$-local quantum Hamiltonian $H$, we
can use powerful heuristic methods to find its quasi-quantum ground
state and quasi-quantum ground energy, since this is essentially a
classical problem. For example, if $M$ is the maximal size of the
support needed by Carath{\'e}odory's theorem to represent a general
sister state, then we may start by randomly picking $M'=M+\Delta$
assignments, where $\Delta=\poly(n)$ is some free parameter we fix.
We then find the minimal energy over this set (which is a SDP
problem that can be solved efficiently), and use the
Carath{\'e}odory's theorem to find the sister state of our optimal
solution. This will reduce the size of the support to $M$. We can
then enlarge the support by adding $\Delta$ new assignments to it,
and re-run the optimization. This will necessarily take us to a
lower energy. Such process is repeated until we hit a local minima.
If such heuristic manages to find (or at least get close) to the
actual qq ground energy, it ought to give us stronger results than
popular relaxation methods that only try to find a set of reduced
density matrices that are consistent with among themselves (see for
example \cc{ref:Lin2022-SDP} and refs within). Unlike these methods,
a qq state is a global object, and therefore the relaxation it
offers is tighter than local relaxation methods.

\section{Background}
\label{sec:background}

\subsection{Quantum states and SIC-POVM measurements}
\label{sec:SIC-POVM}

Given a Hilbert space $\BBH$, we denote the space of linear
operators on it by $L(\BBH)$.  The set of quantum states $\SQ
\subset L(\BBH)$ is a convex subset of $L(\BBH)$ of operators $\rho$
such that (i) $\Tr(\rho)=1$, (ii) $\rho=\rho^\dagger$ and (iii)
$\rho\succeq 0$, i.e., it is a positive semi-definite (PSD)
operator.

A qudit is a quantum register that is described by a $d$-dimensional
Hilbert space $\BBH= \BBC^d$. The Hilbert space of a system of $n$
qudits, is described by the tensor product of the Hilbert spaces of
the individual qudits: $\BBH= (\BBC^d)^{\otimes n}$. Given a
one-qudit operator $Y$ in the multi-qudit setting, we use the
notation $Y^{(i)}$, for $i \in \{1, \ldots, n\}$, when it is
applied on the $i$th qudit of the system. This notation can be
naturally generalized to a constant number of qudits.

Positive Operator Valued Measurements (POVM) model the most general
physical process in which a quantum system is being measured and
then discarded. Mathematically, they are described by a set of
positive semi definite (PSD) operators $\{F_\alpha\}$, where the
index $\alpha$ denotes a possible measurement outcome and in addition
we require the \emph{completeness property}, $\sum_\alpha F_\alpha =
\Id$. The result of a POVM measurement of a quantum state $\rho$ is
an index $\alpha$ from the set of possible results, which in this
work we shall assume to be the set $\{0,1,2,\ldots\}$. The value
$\alpha$ is obtained with probability $\mu(\alpha) = \Tr(F_\alpha
\rho)$. A POVM measurement can therefore be seen as a linear map
that takes a quantum state to a probability distribution
$\mu(\alpha)$.

An informationally complete measurement is a POVM of special type in
which the resultant probability distribution $\mu(\alpha)$ contains
the full information of the measured quantum state $\rho$. This
means that the map $\rho \mapsto \mu(\alpha)$ is invertible and so
$\rho$ can be restored from $\mu(\alpha)$. A necessary and
sufficient condition for a POVM to be informationally complete is
that the POVM operators $\{F_\alpha\}$ span the space of operators
$L(\BBH)$. Therefore, on a quantum system with $\dim(\BBH)=d$, the
minimal number of elements in an informationally complete POVM is
$d^2$. 

Our definition of quasi-quantum states builds on the notion of a
symmetric, informationally complete POVMs, which we now define:
\begin{definition}[SIC-POVM]
\label{def:SIC-POVM} A Symmetric, Informationally-Complete {\rm POVM}
  ({\rm SIC-POVM}) in a $d$-dimensional Hilbert space is an
  informationally complete {\rm POVM} $\{F_\alpha\}_{\alpha=1}^{d^2}$,
  which is specified by $d^2$ unit vectors $\ket{\psi_\alpha}$ such
  that:
  \begin{align}
  \label{eq:F_alpha}
    F_\alpha &\EqDef \frac{1}{d}\ketbra{\psi_\alpha}{\psi_\alpha}, \\
    |\braket{\psi_\alpha}{\psi_\beta}|^2 &= \frac{1}{d+1} \qquad
    \text{for $\alpha\neq \beta$}.
    \label{eq:F-overlap}
  \end{align}
\end{definition}
In other words, a SIC-POVM is an informationally complete
measurement with a minimal number of rank-1 elements, which are also
symmetric with respect to each other. It is straightforward to show
that the $1/d$ normalization in the definition of $F_\alpha$
in~\eqref{eq:F_alpha} and the $1/(d+1)$ normalization in the overlap
in \Eq{eq:F-overlap} follow from the demand that the $F_\alpha$ are
rank-1 operators with the same overlap with each other.

Throughout this work we will mostly be concerned with the case of
qubits, where $d=2$. In that case, we use the SIC-POVM that is given
by the $4$ vectors
\begin{align}
\label{eq:POVM-d2}
  \begin{split}
    \ket{\psi_0} &= \ket{0}, \\
    \ket{\psi_1} &= \frac{1}{\sqrt{3}}\ket{0}
      + \sqrt{\frac{2}{3}}\ket{1},  \\
    \ket{\psi_2} &= \frac{1}{\sqrt{3}}\ket{0} 
      + e^{i2\pi/3}\sqrt{\frac{2}{3}}\ket{1}, \\
    \ket{\psi_3} &= \frac{1}{\sqrt{3}}\ket{0} 
      + e^{i4\pi/3}\sqrt{\frac{2}{3}}\ket{1}. 
  \end{split}
\end{align}
Geometrically, these vectors form a tetrahedral on the Bloch sphere,
as shown in \Fig{fig:SIC-qubit}.
\begin{figure}
  \centering
  \includegraphics{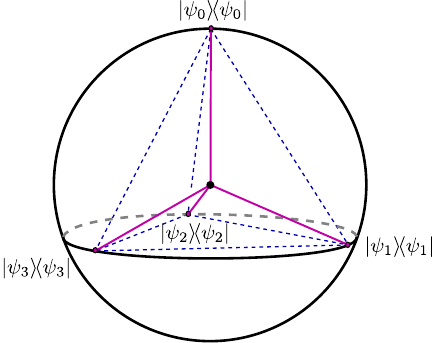}  
  \caption{A common choice for the $4$ elements of a SIC-POVM of
    qubit (see \Eq{eq:POVM-d2}) form a tetrahedral on the Bloch
    sphere.} \label{fig:SIC-qubit}
\end{figure}

Given single-qudit SIC-POVM, we can extend it to multiple qudits
using tensor products. For $n$ qudits, we shall use the notation
$\ualpha=(\alpha_1, \ldots, \alpha_n)$ to denote a string of
possible POVM results for all $n$ qudits, and define
\begin{align}
\label{def:multiE}
  F_\ualpha \EqDef F_{\alpha_1}\otimes F_{\alpha_2}\otimes\ldots
    \otimes F_{\alpha_n}, 
    \qquad \alpha_i \in \{0, 1, \ldots, d^2-1\} .
\end{align}
It is easy to verify that the set $\{F_\ualpha\}$ defines a 
SIC-POVM for the space of $n$ qudits $\BBH = (\BBC^d)^{\otimes n}$.

The multiqudit SIC-POVM basis can be used to map a quantum $\rho$
state to a multivariate probability distribution:
\begin{align}
\label{eq:rho-to-mu}
  \mu(\ualpha) \EqDef \Tr(\rho F_\ualpha) .
\end{align}
Since this is an informationally complete measurement, the map can
be inverted, and we can fully reconstruct $\rho$, given its 
probability distribution $\mu$. An elegant way to do that
is by using the \emph{dual basis} for the SIC-POVM, which is
co-orthogonal to the POVM operators:
\begin{definition}[Dual basis of a SIC-POVM]
\label{def:dual-basis}
  Given a {\rm SIC-POVM} $\{F_\alpha\}$ in a $d$-dimensional space, its
  dual basis is a set of operators $\{D_\alpha\}$ over $\alpha=0,
  \ldots, d^2-1$, given by
  \begin{align}
  \label{eq:Dalpha}
    D_\alpha \EqDef d(d+1)F_\alpha - \Id
      = (d+1)\ketbra{\psi_\alpha}{\psi_\alpha} - \Id .
  \end{align}
\end{definition}
The following properties of the dual basis can be proved by
inspection:
\begin{enumerate}
  \item $\{D_\alpha\}$ is a basis for $L(\BBH)$
  \item $D_\alpha$ is hermitian and $\Tr(D_\alpha)=1$,
  \item $D_\alpha$ is not a PSD operator. It has one positive
    eigenvalue $d$ and $d-1$ eigenvalues of $-1$.
\end{enumerate}
Crucially, the the dual basis is co-orthonormal to the
SIC-POVM basis with respect to the Hilbert-Schmidt inner product:
\begin{align}
\label{eq:co-orthogonality}
  \Tr(F_\alpha D_\beta) = \delta_{\alpha\beta} .
\end{align}
In the case of qubits, for the SIC-POVM in \Eq{eq:POVM-d2}, the
corresponding dual basis is given by
\begin{align}
\label{eq:D2}
  D_\alpha =  3\ketbra{\psi_\alpha}{\psi_\alpha} - \Id
    = 2\ketbra{\psi_\alpha}{\psi_\alpha} 
      -\ketbra{\psi^\perp_\alpha}{\psi^\perp_\alpha}
    \sim \begin{pmatrix}
      2 & \\
      & -1 
    \end{pmatrix} ,
\end{align}
The following lemma lists two useful properties of the dual basis.
\begin{lemma}
\label{lem:D-props}
  The dual basis $\{D_\alpha\}$ defined in \Def{def:dual-basis}
  satisfies the following properties:
  \begin{align}
    \frac{1}{d^2} \sum_\alpha D_\alpha &= \frac{1}{d}\Id , \\
  \text{for every $\beta$}, \quad      
    \frac{1}{d^2-1} \sum_{\alpha\neq \beta} D_\alpha
        &= \frac{1}{d-1}(\Id - \ketbra{\psi_\beta}{\psi_\beta}).
  \end{align}
\end{lemma}
The first equality tells us that, as expected, taking a uniform
distribution over the $D_\alpha$ produces the maximally-mixed sate,
and the second equality tells us that taking the uniform
superposition over all $D_\alpha$ for $\alpha\neq \beta$ gives us
the maximally mixed state in the subspace that is orthogonal to
$\ket{\psi_\beta}$. Since the proof follows almost immediately from
definitions, we omit it.

Just as in the case of a single-qudit POVM $\{F_\alpha\}$, we can
extend the $\{D_\alpha\}$ basis to the multi-qudit setup using
tensor products. For a string $\ualpha=(\alpha_1, \ldots,
\alpha_n)$, we define
\begin{align}
  \label{eq:D_ualpha}
  D_\ualpha \EqDef D_{\alpha_1}\otimes D_{\alpha_2}\otimes\ldots
    \otimes D_{\alpha_n}, \qquad \alpha_i \in \{0,1,2,3\} ,
\end{align}
and notice the co-orthogonality property
\begin{align}
  \Tr(F_\ualpha D_\ubeta) = \delta_{\ualpha, \ubeta} .
\end{align}
With this property, it is easy to reverse the mapping $\rho \mapsto
\mu$. Indeed, expanding $\rho$ in terms of the dual basis and taking
a trace with $F_\ualpha$, we obtain
\begin{align*}
  \mu(\alpha) = \Tr(\rho F_\ualpha) \quad\Leftrightarrow\quad
  \rho = \sum_\ualpha \mu(\ualpha) D_\ualpha .
\end{align*}

Note that the reduced density operator of a multiqudit state $\rho =
\sum_\ualpha \mu(\ualpha) D_\ualpha$ is directly related to the
corresponding marginal of $\mu$. Specifically, for a set of $k$
qudit labels $I=(i_1, i_2, \ldots, i_k)$, denote the reduced density
operator of $\rho$ on $I$ by $\rho_I \EqDef \Tr_{\setminus I} \rho$, and
the corresponding marginal of $\mu$ by $\mu_I$. Then from the fact
that $\Tr(D_\alpha)=1$ for every $\alpha$, we conclude that
\begin{align}
\label{eq:RDM-I}
  \rho_I &= \sum_{\ualpha_I} \mu_I(\ualpha_I) 
    D_{\alpha_1}\otimes\ldots \otimes D_{\alpha_n} , &
    \ualpha_I &\EqDef (\alpha_{i_1}, \ldots, \alpha_{i_k}).
\end{align}

At this point it is worthwhile to stress that although every quantum
state gives rise to a probability distribution $\mu(\ualpha)$ via
the SIC-POVM measurement, the converse is not true. If
$\mu(\ualpha)$ is an arbitrary probability distribution, then
generally $\sum_\ualpha \mu(\ualpha)D_\ualpha$ will \emph{not} be a
PSD operator, and hence not a quantum state. For example, in the
single qubit case, if we take $\mu(\alpha) = \delta_{\alpha,0}$,
then $\sum_\alpha \mu(\alpha) D_\alpha = D_0$, which, as we have
observed has a negative eigenvalue $-1$, and cannot be a quantum
state.

\section{\texorpdfstring{$k$}{k}-local quasi-quantum states 
 (\texorpdfstring{$k$}{k}-qq states)}
\label{sec:qq-states}

We are now ready to define the main objects of this work:
$k$-local quasi quantum states, or in short $k$-qq states. We start
with a formal definition of these states, and then we study some
their properties.

\subsection{Definition of \texorpdfstring{$k$}{k}-local 
  quasi-quantum states}
\label{sec:qq-def}

$k$-local quasi-quantum states over $n$ qudits are a convex set of
$n$ qudits operators, which is a superset of the quantum states. It
is defined with respect to a fixed locality parameter $k=1,2,\ldots$
and a SIC-POVM $\{F_\alpha\}$ as follows. 
\begin{definition}[$k$-local quasi-quantum states]
\label{def:qq-state}
  Let $\BBH = \big(\BBC^d\big)^{\otimes n}$ be the Hilbert space of
  $n$ qudits. An operator $X\in L(\BBH)$ is called a
  $k$-local quasi-quantum state ($k$-local qq for short) if it
  satisfies the following $3$ properties:
  \begin{description}
      
    \item [1. Normalization:] $\Tr(X) = 1$
    
    \item [2. Hermiticity:]  $X^\dagger = X$ 
    
    \item [3. Positivity:] {~}
      \begin{description}
        \item [3a. Global positivity with respect to the SIC-POVM:]\ \\
          For every $\ualpha=(\alpha_1, \ldots, \alpha_n)$, 
          \begin{align}
            \Tr(F_\ualpha X) \ge 0 .
          \end{align}
      
        \item [3b. Local positivity:] \ \\
          for every set of indices $I=\{i_1, i_2, \ldots, \}$ with
          $|I| = k$, define $X_I \EqDef \Tr_{\setminus I} X$ to be
          the ``marginal'' of $X$ on $I$. Then
          \begin{align*}
            X_I \succeq 0 .
          \end{align*}
      \end{description}
  \end{description}  
  We denote the set of all $k$-local qq states by $\Sqq$.
\end{definition}

From the above definition, it is clear that $\Sqq$ is a convex set
of operators that includes the set of quantum states. The global
positivity condition~(3a), as well as the Hermiticity and trace $1$
conditions, imply that, just like ordinary quantum states, we
can use the SIC-POVM $\{F_\ualpha\}$ to map $X$ to a probability
distribution $\mu$ by:
\begin{align}
  \mu(\ualpha) \EqDef \Tr(XF_\ualpha) .
\end{align}
An equally good way of viewing qq states is therefore as probability
distributions over the set $[d^2]^n$, which are subject to a set of
convex constraints on their $k$-local marginals. We can make this
point of view more explicit using the dual basis $\{D_\ualpha\}$
(see \Def{def:dual-basis}). As we have seen, given a probability
distribution $\mu$,  
\begin{align*}
  \mu(\ualpha) = \Tr(XD_\ualpha) \quad \Leftrightarrow\quad
  X = \sum_\ualpha \mu(\alpha)D_\ualpha.
\end{align*}
We shall therefore often
use the notation
\begin{align}
\label{def:X-mu}
  X(\mu) \EqDef \sum_\ualpha \mu(\alpha)D_\ualpha 
\end{align}
to denote an operator whose underlying distribution is $\mu$. When
$\mu$ is a point distribution $\mu(\ualpha^*) = 1$ for some
$\ualpha^*$, we write $X(\ualpha^*)$ for $X(\mu)$. With this
notation, we arrive to the following equivalent definition of a
$k$-local qq state:
\begin{definition}[An equivalent definition of $k$-local
  quasi-quantum states]
\label{def:alt-def}
  An operator $X$ is a $k$-local quasi-quantum state over $n$ qudits
  if and only if there exists a probability distribution
  $\mu(\ualpha)$ such that
  \begin{align}
    X = X(\mu)
  \end{align}
  and $X_I\succeq 0$ for any set of $k$ indices $I=(i_1, \ldots,
  i_k)$.
\end{definition}

As in the case of regular quantum states (see \Eq{eq:RDM-I}), the
$k$-local marginal of a qq-state is directly related to the
corresponding marginal of $\mu$: if $I=(i_1, \ldots, i_k)$ is a set
of $k$ qudit labels, then 
\begin{align}
\label{eq:X-I}
  X_I(\mu) = \sum_{\ualpha_I} \mu_I(\ualpha_I) 
    D_{\alpha_1}\otimes\ldots\otimes D_{\alpha_n} .
\end{align}

By definition, $\Sqq$ is closed with respect to convex combinations.
Other operations that preserve it are tracing out and tensor
products: if $X$ is a $k$-local qq state over $n$ qudits, and $I$ is
a subset of $\ell<n$ qudits, then also $X_I = \Tr_{\setminus I} X$
is a $k$-local qq state. In addition, if $X_1, X_2$ are two
$k$-local qq states then so it $X_1\otimes X_2$.  Other than that, 
there are not many other natural operations that preserve it.

\subsection{The support of a \texorpdfstring{$k$}{k}-qq state}

An important characteristic  of a $k$-qq state is its \emph{support}:
\begin{definition}[Support of a qq state]
  Given a $k$-qq state $X(\mu)$, its support is the support of the
  probability distribution $\mu$, i.e.,
  \begin{align*}
    \supp(X) \EqDef \{\ualpha | \mu(\ualpha)>0\} 
      = \{\ualpha | \Tr(F_\ualpha X)>0\}.
  \end{align*}
\end{definition}
qq states with a polynomial support are interesting from a
computational complexity point of view since they posses an
efficient classical description. It might therefore come as a little
surprise that genuine quantum states, even pure ones, should have an
exponential support. For example, by \Lem{lem:D-props}, we know that
in the case of qubits, 
\begin{align*}
  \ketbra{1}{1} = \frac{1}{3} \big(D_1 + D_2 + D_3\big) .
\end{align*}
and so the pure quantum state $\ketbra{1}{1}^{\otimes n}$ has
exactly $3^n$ strings in its support (all the strings that do not
contain $\alpha_i=0$). The following lemma, whose proof is given in 
\App{app:q-support}, shows that the above example is optimal and that 
\emph{any} quantum state has an exponential support.
\begin{lemma}
\label{lem:q-support} If $X\in \Sqq$ is a quantum state (i.e.,
$X\succeq 0$) on $n$ qubits then $|\supp(X)| \ge 3^n$. 
\end{lemma}

{~}

However, if we consider qq states, and all we are interested
in are local expectation values, then as the following construction
shows, efficient qq states are sufficient: for every $X\in \Sqq$
there exists $X'\in \Sqq$ with a \emph{polynomial} support (and
therefore an efficient classical representation) that has the same
$k$-local marginals as $X$. 
\begin{theorem}[Sister states]
\label{thm:sister}
  For every $k$-local qq state $X$ there exists a \emph{sister
  state} $X'$ such that:
  \begin{enumerate}
    \item For every $k$-local subset of qudits $I$,
      we have $X_I = X'_I$.
    \item $|\supp(X')| = O(n^k\cdot d^{2k})$.
  \end{enumerate}
\end{theorem}

\begin{proof}
  The proof of \Thm{thm:sister} follows the ideas of
  \cc{ref:Pitowsky1991-CorrPolytope}, which used
  Carath{\'e}odory's theorem\cc{ref:Caratheodory1911} to show that
  the marginal problem is in NP.
  
  Let $X(\mu)=\sum_\ualpha \mu(\ualpha) D_\ualpha$ be a $k$-local
  qq-state. By \Eq{eq:X-I}, the $k$-local marginal $X_I(\mu)$ is
  given by the corresponding $k$-local marginal of $\mu$, and
  therefore it is sufficient to show that for every multivariate
  distribution $\mu$, there exists a sister distribution $\mu'$ that
  has the same $k$-local marginals but with a $O(n^k\cdot d^{2k})$
  support. 
  
  To do that, we first map $\mu$ to the space of marginals. The
  distribution $\mu$ can be represented as a vector in $\BBR^L$ with
  non-negative entries, where $L=(d^2)^n$. To stress the fact that
  we now view it as a vector, we shall slightly abuse the Dirac
  notation and denote it by $\ket{\mu}$. Similarly, there
  are $\binom{n}{k}$ different $k$-local marginals $\mu_I$, which
  can all be arranged in a vector $(\mu_{I_1}, \mu_{I_2}, \ldots)$.
  Each of these marginals can be represented as a vector in
  $\BBR^{(d^2)^k}$, and therefore the set of all $k$-local marginals
  can be represented as a vector in $\BBR^M$, with
  $M=\binom{n}{k}\cdot (d^2)^k = O(n^k\cdot d^{2k})$.  We denote this mapping by
  $\mcM:\BBR^L \to \BBR^M$, and note that it is a linear map. Using
  the Dirac notation once more, the vector of marginals is denoted
  by $\ket{\mcM(\mu)}$. 
    
  Next, for every string $\ubeta=(\beta_1, \ldots, \beta_n)$, 
  denote by $\delta_\beta$ the sharp distribution
  \begin{align*}
    \delta_\beta(\ualpha) \EqDef 
    \begin{cases}
      1, & \ualpha=\ubeta, \\
      0, & \text{otherwise} 
    \end{cases} .
  \end{align*}
  Using our Dirac notation, we can write $\ket{\mu} = \sum_\ubeta
  \mu(\ubeta) \ket{\delta_\ubeta}$, and by the linearity of $\mcM$,
  it follows that $\ket{\mcM(\mu)} = \sum_\ubeta \mu(\ubeta)
  \ket{\mcM(\delta_\ubeta)}$. By definition, $\ket{\mcM(\mu)}$ is a
  vector in the convex hull of the vectors
  $\{\ket{\mcM(\delta_\ubeta)}\}_\ubeta$ in $\BBR^M$. Therefore, by
  Carath{\'e}odory's theorem\cc{ref:Caratheodory1911}, it can be
  written as a convex combination of at most $M+1$ of these vectors:
  \begin{align*}
    \ket{\mcM(\mu)} 
      &= \sum_{i=1}^{M+1} p_i \ket{\mcM(\delta_{\ubeta_i})} .
  \end{align*}
  It follows that the distribution $\mu'$, whose vectorial
  representation is
  \begin{align*}
    \ket{\mu'} \EqDef \sum_{i=1}^{M+1} p_i\ket{\delta_{\ubeta_i}}
  \end{align*}
  has the same marginals as $\mu$ and $|\supp(\mu')|\le M+1 =
  O(n^k\cdot d^{2k})$.
\end{proof}

\section{The \texorpdfstring{$k$}{k}-local Hamiltonian problem 
  over quasi-quantum states}
\label{sec:k-LH}

In this section we present the $k$-local Hamiltonian problem, and
define it in the context of quantum states and quasi-quantum states.
Using the sister states (\Thm{thm:sister}), we show that this
problem is inside NP. In the next section we shall show that it is
also NP-hard, thus making it an NP-complete problem.

Given a system of $n$ qubits $\BBH = (\BBC^2)^{\otimes n}$, a
$k$-local Hamiltonian is a Hermitian operator $H\in L(\BBH)$ given
as $H=\sum_{i=1}^m h_i$, where each $h_i$ is a Hermitian operator
with $\norm{h_i}=O(\poly(n))$ that acts non-trivially on at most $k$
qubits and $m\le \binom{n}{k}\le n^k$. In other words, each $h_i$
can be written as $h_i=\hat{h}_i\otimes\Id_{rest}$, where
$\hat{h}_i$ is a Hermitian operator on the space of $k$ qubits, and
$\Id_{rest}$ is the trivial operation on the rest of the qubits in
the system.

For any local Hamiltonian $H$, we define its ground energy and
ground states as follows:
\begin{align*}
  \eps_0 &\EqDef \min_{\rho\in \SQ} \Tr(H\rho) , \\
  \rho_0 &\EqDef \minarg_{\rho\in \SQ} \Tr(H\rho) ,
\end{align*}
{where $\SQ$ is the set of all quantum states.} As $\SQ$ is a convex
set and $\Tr(H\rho)$ is a linear function, there always exist a
minimum at the boundary of $\SQ$, which is a pure state $\rho_0 =
\ketbra{\gs}{\gs}$. Therefore, $\eps_0 = \min_{\ket{\psi}}
\bra{\psi}H\ket{\psi}$ and $\ket{\gs} =\minarg_{\ket{\psi}}
\bra{\psi}H\ket{\psi}$. This means that $\eps_0$ is the minimal
eigenvalue of $H$ and $\ket{\gs}$ is an eigenstate of $H$ with a
minimal eigenvalue. When considering approximations to the ground
energy, it is natural to express them in terms of a scale, which we
define below
\begin{definition}[The approximation scale $L_H$]
\label{def:L-H}
  Given a local Hamiltonian $H=\sum_i h_i$, we define its
  approximation scale $L_H$ by
  \begin{align}
    \label{def:L-H-formula}
    L_H \EqDef \sum_i \norm{h_i} .
  \end{align}
\end{definition}
We note that $L_H$ can be viewed as an efficiently calculable
proxy to the Hamiltonian norm $\norm{H}$, and that $\norm{H}\le
L_H\le \poly(n)$ and therefore $-L_H \le \eps_0 \le L_H$.

The $k$-local Hamiltonian  problem is a promise problem that
formalizes the complexity of approximating $\eps_0$:
\begin{problem}[The $k$-local Hamiltonian problem]
\label{prob:qLH} 
  Given a $k$-local Hamiltonian $H=\sum_i h_i$ over $n$ qubits and
  two real numbers $a<b$, the $k$-$\LH(H,a,b)$ problem is the
  following promise problem. We are promised that the ground energy
  $\eps_0$ of $H$ is either $\eps_0\le a$ (the {\rm YES} case) or
  $\eps_0\ge b$ (the {\rm NO} case). We have to decide which case
  holds.
\end{problem}

A seminal result by Kitaev is that the $k$-local Hamiltonian problem
is QMA-complete\cc{ref:Kitaev1999-LH, ref:Kitaev2002-QI} for $k=5$
and $b-a=\Theta(L_H/\poly(n))$. This result is often considered to
be the quantum analog of the famous Cook-Levin
theorem\cc{ref:Cook71, ref:Levin73}. Kitaev's result has been
considerably strengthened in several ways.  For example, it has been
shown to remain QMA-hard when $k=2$\cc{ref:Kempe2006-LH}, also
when the interactions are described by a planar 
graph\cc{ref:Oliveira2008-2D}. It was even proved to be QMA-hard for
1D systems with nearest-neighbor $h_i$ that act on qudits of
dimension $d=8$\cc{ref:Aharonov2009-1D, ref:Nagaj2008-PhD,
ref:Hallgren2013-1D}. All of these results hold for a promise gap
$b-a = \Theta(L_H/\poly(n))$.

What happens to the complexity of the problem when the promise gap
increases? By the classical PCP theorem, there exists a constant
$\xi >0$ such that the $k$-$\LH(H,a,b)$ problem for a classical
Hamiltonian $H$ becomes NP-hard for $b-a=\xi\cdot L_H$. Indeed,
taking $H$ to be a classical Hamiltonian, where each $h_i$ is a
projector diagonal in the computational basis, the $k$-$\LH(H,a,b)$
problem with $a=0$ becomes the promise version of the MAX-$k$-SAT
problem, for which there is a hardness of approximation result by
the PCP theorem\cc{ref:Arora1998-PCP1, ref:Arora1998-PCP2,
ref:Dinur2007-PCP}.

But what is the complexity of the $k$-$\LH(H,a,b)$ problem for
\emph{general} $k$-local Hamiltonians when $b-a=\Theta(L_H)$? Does
it remain QMA-hard, or does it lose its hardness and becomes, for
example, an NP problem? This is the question of the quantum PCP
conjecture (in its Hamiltonian formulation), which speculates that
the problem remains QMA-hard in that limit. Unfortunately, despite
almost two decades of research, this question remains wide open.
There are several results that support this conjecture, like a (very
restricted) type of gap amplification
routine\cc{ref:Aharonov2009-DL}, or the proof of the NLTS
conjecture\cc{ref:Anurag2023-NLTS}, but also results that suggest
its incorrectness such as the existence of product-state
approximations for a wide range of low-energy
states\cc{ref:Brandao2013-prod-states}, or examples where the
complexity of a quantum optimization problem changes its hardness as
the promise gap changes from $1/\poly(n)$ to a
constant\cc{ref:Aharonov2019-stoqPCP,
ref:Gharibian2022-dequant-QSVT}. See \cRef{ref:Buherman2024:qPCP}
for a recent results.

The aim of the current work is to make progress on this question by
studying the $k$-$\LH(H, a,b)$ problem over the \emph{$k$-local
quasi-quantum states} as a toy-model for the actual quantum problem.
As we shall show, over these states the complexity of the $k$-$\LH$
problem does \emph{not} change; it remains NP-complete for promise
gaps ranging from $b-a=L_H\cdot e^{-O(n)}$ to $b-a=\Theta(L_H)$.

We begin by formalizing the $k$-local Hamiltonian problem over the
set of quasi-quantum states. Going from $\SQ$ to $\Sqq$, we define
\begin{definition}[The quasi-quantum ground energy and ground state]
  Given a $k$-local Hamiltonian $H=\sum_i h_i$, its  quasi-quantum
  ground energy $\eps_0^{qq}$ and ground state $X_0$ are defined by
  \begin{align}
  \label{def:eps0-qq}
    \eps^{qq}_0 &\EqDef \min_{X\in \Sqq} \Tr(H X) , \\
    X_0 &\EqDef \minarg_{X\in \Sqq} \Tr(HX) .
  \label{def:X0-qq}
  \end{align}
\end{definition}
 Note that we use the \emph{same} parameter $k$ for the locality of
  $H$ and the locality of the qq sates.
  
As a quick remark, note that $\Tr(HX) = \sum_i \Tr(h_iX) =
\sum_i\Tr(h_i X^{(i)})$, where we slightly abuse the notation and
let $X^{(i)}$ be the marginal of $X$ on the non-trivial qubits of
$h_i$. Since $X$ is a $k$-local qq state, it follows that
$X^{(i)}\succeq 0$ and $\Tr(X^{(i)})=1$, and therefore $|\Tr(h_i
X^{(i)})|\le \norm{h_i}$, from which we deduce that $|\eps^{qq}_0| =
|\sum_i \Tr(X^{(i)}_0 h_i)| \le \sum_i\norm{h_i} = L_H$ is
well-defined. Moreover, as $\SQ\subset \Sqq$, it follows that
$\eps_0^{qq} \le \eps_0$.

Having defined the quasi-quantum ground energy, we now define the
$k$-local Hamiltonian problem over the quasi-quantum states in
parallel with Problem~\ref{prob:qLH}:
\begin{problem}[The quasi-quantum $k$-local Hamiltonian problem]
\label{prob:qqLH} 
  Given a local Hamiltonian $H=\sum_{i=1}^m h_i$
  over $n$ qubits, and two real numbers $a<b$, the
  $k$-$\qqLH(H,a,b)$ problem is the following promise problem. We
  are promised that the qq ground energy $\eps^{qq}_0$ of $H$ is
  either $\eps^{qq}_0\le a$ (the {\rm YES} case) or
  $\eps^{qq}_0\ge b$ (the {\rm NO} case). We have to decide which
  case holds.  
\end{problem}

An easy corollary of \Thm{thm:sister} is that every qq state $X$ 
has a sister state $X'$ with a polynomial support and the same
marginals. Since the quasi-quantum energy $\Tr(XH)$ of a $k$-local
Hamiltonian only depends on the $k$-local marginals, it follows that
$X'$ and $X$ have the same energy. This leads us to the following
theorem, which places the $k$-local Hamiltonian problem over
quasi-quantum states inside NP:
\begin{theorem}
\label{thm:inside-NP}
  The quasi-quantum $k$-local Hamiltonian problem 
  $k$-$\qqLH(H,a,b)$ (Problem~\ref{prob:qqLH}) is inside {\rm NP}
  for $b-a = \Omega(L_H\cdot 2^{-O(n)})$.
\end{theorem}
\begin{proof}
  Given a $k$-local Hamiltonian $H=\sum_i h_i$ and two constants
  $a<b$, the verification protocol is as follows: the prover sends
  the verifier an efficient description of the quasi-quantum ground
  state in the form of the probability distribution $\mu(\alpha)$
  with a polynomial support\footnote{In fact, it is sufficient for
  the prover just to send $\supp(\mu)$, since the minimal energy
  over that set can then be efficiently computed by the verifier
  using a semi-definite program.}.  The verifier calculates all the
  marginals of this operator and verifies that they are all PSD so
  that the witness is indeed a $k$-local quasi-quantum state. He
  uses these marginals to calculate the energy and accept if it is
  below $a$. By definition the protocol is complete. It is also
  sound, because if the ground energy is $\ge b$, and the marginals
  of the prover are PSD, then its witness is a qq state, and hence
  the energy calculated by the verifier will be $\ge b$. Using
  standard techniques, all of the above calculations can be done
  efficiently to within exponential precision using a description of
  the qq state with $\poly(n)$ bits.
\end{proof}

\section{The quasi-quantum PCP theorem}
\label{sec:NP-hardness}

In this section we show that $k$-$\qqLH(H,a,b)$ problem
(Problem~\ref{prob:qqLH}) remains NP-hard for $b-a = \Theta(L_H)$,
thereby proving a PCP hardness-of-approximation of the local
Hamiltonian problem over quasi-quantum states. Our proof uses a
reduction from 4-coloring. Specifically, we shall prove the
following theorem:
\begin{theorem}[The quasi-quantum PCP theorem]
\label{thm:main-NP-hard}
  There exists constants $0<\xi_a<\xi_b<1$ such that the problem
  $3$-$\qqLH(H, \ \xi_a\cdot L_H, \ \xi_b\cdot L_H)$ is
  {\rm NP}-hard. 
\end{theorem}

In the next subsections we prove this theorem. We begin by an
introductory section that discusses the 4-coloring problem in the
context of the qq states and local Hamiltonians. As we shall see, on
one hand, 4-coloring is a natural problem for qq states over qubits,
but on the other hand, it cannot be straightforwardly applied to
prove NP-hardness. After that, we give a bird-eye view of our
solution, which uses the so-called scapegoat mechanism and
$\lambda$-solutions. In the rest of the subsections, we give our
full proof.

\subsection{The 4-coloring problem over qq states}
\label{sec:4-coloring}

Consider the vertex 4-coloring problem over a graph $G=(V,E)$ with
$n=|V|$ vertices, asking whether $G$ has a legal 4-coloring. By
definition, a $4$-\emph{coloring} is a mapping of the vertices into
$\{0,1,2,3\}$, called the colors, and it is called \emph{legal} if
the two endpoints of every edge have different colors.  This problem
can be naturally mapped to a 2-local Hamiltonian over 2-local qq
states as follows. We first put a qubit on every vertex. Then any
string $\ualpha=(\alpha_1, \ldots, \alpha_n) \in \{0,1,2,3\}^n$ of
the POVM basis can be regarded as a possible coloring of the
vertices. To force the coloring to be legal, we penalize neighboring
qubits that share the same color. We do this using the Hamiltonian
\begin{align}
\label{def:H_G}
  H_G &\EqDef \sum_{e\in E} h_e, \\
    h_e &\EqDef \sum_{\alpha=0}^3 F_\alpha^{(i)} \otimes F_\alpha^{(j)} 
  \quad \text{for every $e=\{i,j\}\in E$} .
\label{eq:coloring-h}
\end{align}
This way, 
\begin{align*}
  \Tr(h_e D_{\alpha}\otimes D_\beta) 
  = \begin{cases}
      0 ,& \alpha\neq \beta \\
      1 ,& \alpha =\beta 
    \end{cases}
\end{align*}
and therefore for every string $\ualpha$,
\begin{align*}
  \Tr(H_G D_{\ualpha}) = \text{number of coloring violations of
  $\ualpha$ in $G$} .
\end{align*}
If $G$ is not 4-colorable, $\Tr(H_GD_\ualpha)\ge 1$ for each
$\ualpha$.  In such case, for every distribution $\mu$ over
4-colorings, we have $\Tr(H_GX(\mu)) \ge 1$, where we recall that
$X(\mu)=\sum_\ualpha \mu(\ualpha) D_\ualpha$ (see \Eq{def:X-mu}).

Interestingly, an explicit calculation shows that $h_e =
\frac{1}{3}(\Id-\ketbra{\psi_-}{\psi_-})$, where $\ket{\psi_-} =
\frac{1}{\sqrt{2}}(\ket{01}-\ket{10})$, which can also be written as 
\begin{align*}
  h_e = \frac{1}{12}(\sigma_X\otimes \sigma_X 
    + \sigma_Y\otimes\sigma_Y + \sigma_Z\otimes \sigma_Z) 
    + \frac{1}{4}\Id, 
\end{align*}
where $\sigma_X, \sigma_Y, \sigma_Z$ are the Pauli matrices. This
is, up to an overall shifting by unity and scaling of the
Hamiltonian, the anti-ferromagnetic Heisenberg Hamiltonian (AFH).

The above construction shows that the mapping of the 4-coloring
problem to the $2$-local Hamiltonian problem is sound. However,
completeness is broken. Indeed, assume that $G$ is $4$-colorable,
and let $\ualpha^*= (\alpha_1, \alpha_2, \ldots, \alpha_n)$ be a
legal 4-coloring. Then $\Tr(H_G D_{\ualpha^*})=0$, however
$D_{\ualpha^*}$ \emph{is not a $2$-local qq state} --- every
marginal of it has some negative eigenvalues.

We can partially tackle this problem by \emph{symmetrizing} the
solution. The negativeness of our solution stems from it being made
of a single $D_\ualpha$ operator. To avoid negativity, we want to
``cushion'' the negativity in $D_\ualpha$ by taking a convex
combination of many such solutions. This can be done
by considering all the $4!=24$ color permutations of the set of 4
colors $\{0,1,2,3\}$ for a specific coloring.  Observe that if
$\ualpha^*=(\alpha_1, \alpha_2, \ldots, \alpha_n)$ is a legal
coloring of some graph $G=(V,E)$, then so is $\pi(\ualpha^*) \EqDef
\big(\pi(\alpha_1), \pi(\alpha_2), \ldots, \pi(\alpha_n)\big)$,
where $\pi(\cdot)$ denotes any one of the 24 color permutations. We
therefore consider the the \emph{symmetric coloring operator} for
specific colorings:

\begin{definition}[Symmetric coloring operator]
\label{def:symm-coloring operator} 
  Let $\ualpha^*$ be a coloring of $G$. The symmetric-coloring
  operator of $\ualpha^*$ is
  \begin{align}
    X_\sym(\ualpha^*) 
      \EqDef \frac{1}{24}\sum_\pi D_{\pi(\ualpha^*)} ,
  \end{align}
  where $\sum_{\pi}$ denotes the sum over all $24$ possible color
  permutations. 
\end{definition}
Unfortunately, as we show below, also $X_\sym(\ualpha^*)$ is not a
proper 2-local qq state, and this is why finding a sound reduction
from 4-coloring to the qq LH problem is a not an easy task. To see
this, we first define a large family of operators to which the
symmetric coloring operators belong:
\begin{definition}[Color-invariant distribution and operator]
\label{def:color-inv} 
  A distribution $\mu(\ualpha)$ is color invariant if for every
  $\ualpha$ and a color permutation $\pi$, we have $\mu(\ualpha) =
  \mu\big(\pi(\ualpha)\big)$. For a color-invariant distribution
  $\mu$, we call the operator $X(\mu) = \sum_\ualpha \mu(\ualpha)
  D_\ualpha$ also color invariant.
\end{definition}

To describe the 2-local marginals of color-invariant operators, we
define the following two-qubits operators:
\begin{align}
  \label{def:SA}
    S &\EqDef \frac{1}{4}\sum_{\alpha=0}^3 D_\alpha\otimes D_\alpha, &
    A & \EqDef \frac{1}{12}\sum_{\alpha\neq \beta} 
      D_\alpha\otimes D_\beta.
\end{align}
Using \Lem{lem:D-props}, we deduce that $\sum_\alpha D_\alpha =
2\Id$, and therefore $4S + 12A = \sum_{\alpha,\beta} D_\alpha\otimes
D_\beta = 4\Id$. This gives us $S = \Id - 3A$. Then a direct calculation
shows that $A = \ketbra{\psi_-}{\psi_-}$, where $\ket{\psi_-}\EqDef
\frac{1}{\sqrt{2}} (\ket{01}-\ket{10})$ is the singlet state, and
therefore,
\begin{align}
\label{eq:A}
  A &= \ketbra{\psi_-}{\psi_-}, \\
  S &= \Id - 3\ketbra{\psi_-}{\psi_-} .
\label{eq:S}
\end{align}
We see that while $A$ is a rank-1 projector with eigenvalues
$(1,0,0,0)$, $S$ has eigenvalues $(1,1,1,-2)$ and is therefore not a
PSD operator. The following lemma uses this to relate the the form
of the local marginal in color invariant operators to the 
\emph{collision probability} $\tau_\mu^{(ij)} \EqDef
\Prob_\mu(\alpha_i=\alpha_j)$, which is the probability that the
colors at vertices $i,j$ are the same, when the colors are
distributed according to $\mu$.
\begin{lemma}
\label{lem:2marginals} 
  Let $X(\mu)$ be a color-invariant operator. Then for every $i,j
  \in \{1, \ldots, n\}$, its $1$-local and $2$-local marginals are
  given by
  \begin{align}
  \label{eq:1-marginal}
    X^{(i)}&= \frac{1}{2} \Id , \\
    X^{(ij)} &= \tau_{\mu}^{(ij)} S + (1-\tau_{\mu}^{(ij)}) A
    = \tau_{\mu}^{(ij)} \Id 
      + (1-4 \tau_{\mu}^{(ij)} )\ketbra{\psi_-}{\psi_-} ,
    \label{eq:SA-marginal}
  \end{align}
  where $ \ket{\psi_-}\EqDef \frac{1}{\sqrt{2}}(\ket{01}-\ket{10})$
  is the singlet state.
\end{lemma}

\begin{proof}
  Equation~\eqref{eq:1-marginal} and the first equality in
  \Eq{eq:SA-marginal} follow directly from definition. The
  second equality in \Eq{eq:SA-marginal} follows from
  Eqs.~(\ref{eq:A},~\ref{eq:S}).
\end{proof}
From \Eq{eq:SA-marginal} we get the following corollary.
\begin{corollary}
\label{corol:lambda-min}
  Let $X(\mu)$ 
  be a color invariant operator, and for any $i,j$, 
  let $ \lambda_{\mu}^{(ij)}$ denote 
  the minimal eigenvalue  of $X_{ij}$.  Then
  \begin{align*}
    \lambda_{\mu}^{(ij)} 
        = \min(\tau_{\mu}^{(ij)}, 1-3\tau_{\mu}^{(ij)}),
  \end{align*}
  and therefore $\lambda_{\mu}^{(ij)} > 0$ if and only if
  $0<\tau_\mu^{(ij)}<1/3$. 
 \end{corollary}
Therefore, showing the positiveness of local marginals in
color-invariant operators is equivalent to showing that the
collision probability $\tau_\mu^{(ij)}$ in the underlying
distribution $\mu$ is in $(0,1/3)$.

We are now in position to explain why a symmetric-coloring operator
$X(\ualpha^*)$ is not always a 2-local qq state. Indeed if we have
more than $4$ vertices in the system, then necessarily there are two
vertices $i,j$ that share the same color with probability $1$. The
collision probability between these two vertices is
$\tau_\mu^{(ij)}=1$, hence their marginal is $X^{(ij)}(\ualpha^*) =
\frac{1}{4}\sum_{\alpha=0}^3 D_\alpha\otimes D_\alpha = S$ ---
which, as we saw, is not a PSD operator.

To summarize, there is a direct and natural mapping of the
4-coloring problem to a 2-local qq LH problem. This reduction is
trivially sound, but is not complete. If $G$ has a valid 4-coloring,
it is not clear how to turn it into a valid 2-local qq state with
vanishing energy.  Using symmetrization, we can elevate this problem
by creating a color-invariant operator, which is a convex
combination of the 24 permutations of a given coloring. But while
this construction gives PSD 1-local marginals, it fails to produce
2-local marginals that are PSD.  Intuitively, to solve this problem
our qq state must contain many more coloring strings $D_\ualpha$,
which together lead to local positivity. In the next subsection, we
give a bird's-eye view of how we solve this problem by moving to a
$k=3$ local Hamiltonian problem, in conjunction with the
\emph{scapegoat mechanism}. In the rest of the subsections, we prove
this formally.

\subsection{The scapegoat mechanism}
\label{sec:bird-eye}

As we have seen in the previous subsection, given a legal
$4$-coloring $\ualpha^*$, the symmetric operator $X_\sym(\ualpha^*)$
is not locally positive and is therefore not a legal qq state. An
easy way to fix that, would be to add it as a small perturbation to
a legal qq state that is \emph{strictly} positive in every marginal,
for example, the maximally mixed state $X_\mathrm{max}\EqDef
\frac{1}{2^n}\Id$. For every pair of vertices,
$X_\mathrm{max}^{(ij)} = \frac{1}{4}\Id$, and therefore it is clear
that for a sufficiently small (yet constant) $\eps$, the state
$\hat{X}\EqDef (1-\eps)X_\mathrm{max} + \eps X_\sym(\ualpha^*)$ is
locally positive. The problem with this suggestion is, of course,
that the energy of $\hat{X}$ will no longer be zero. In fact, as
$\Tr(X_\mathrm{max} h_e) = \frac{1}{4}$ for every coloring
constraint $h_e$, we expect the energy of $\hat{X}$ to be
$\Omega(|E|)$, which does not yield a promise gap. 

To solve this, we use an idea that we call a ``scapegoat''. It is an
ancillary qubit that we add to the system on which we
\emph{condition} the coloring constraints. We shall refer to the
original $n$ qubits of the problem as \emph{bulk qubits}. Just as in
biblical times, when a single poor goat was released to the
wilderness, carrying the sins of an entire nation\cc{ref:scapegoat},
the scapegoat qubit allows us to shift the energy penalties of the
many coloring constraints to a single qubit, thereby drastically
decreasing the overall energy penalty. Specifically, we define two
possible states of the scapegoat qubit:
\begin{definition}[The $A_0, A_1$ operators]
\label{def:Ai} 
  The single-qubit operators $A_0, A_1$ are defined by
  \begin{align}
    A_0 &\EqDef D_0 &&= 2\ketbra{0}{0} - \ketbra{1}{1} 
    &\text{(inactive scapegoat)}, \\
    A_1 &\EqDef \frac{1}{3}(D_1+D_2+D_3) &&= \ketbra{1}{1} 
    &\text{(active scapegoat)} .
  \end{align}
\end{definition}
When the scapegoat is in state $A_1$, we say that it is active, and
when it is in state $A_0$, it is inactive. Accordingly, we modify
the coloring constraint $h_e$ by conditioning it on the scapegoat
qubit being inactive. Defining the index of the scapegoat qubit to
be $n+1$, we extend $h_e = \sum_{\alpha=0}^3 F_\alpha^{(i)}\otimes
F_\alpha^{(j)}$ into a 3-qubit constraint including also the
scapegoat qubit 
\begin{align}
  h_e \otimes F^{(n+1)}_0.
\end{align}
This way, the ``scapegoated'' coloring constraint will give an
energy penalty only when $\alpha_i=\alpha_j$ \emph{and} the
scapegoat is inactive (in state $A_0=D_0$). To preserve soundness,
we add to the Hamiltonian a term $\Id-F_0^{(n+1)}$, which penalizes
configurations in which the scapegoat qubit is active. Overall, our
modified Hamiltonian is now $3$-local and is defined on $n+1$ qubits
by
\begin{align*}
  H' = \sum_{e\in E} h_e \otimes F_0^{(n+1)} + (\Id-F_0^{(n+1)}).
\end{align*}
It is easy to see that if the graph is non colorable, the energy of
any $n+1$ string $D_\ualpha$ (and therefore of any $3$-local qq
state) will be $\ge 1$. On the other hand, when the graph is
4-colorable and $\ualpha^*$ is a legal coloring, we can take our
original state $\hat{X}\EqDef (1-\eps)X_\mathrm{max} + \eps
X_\sym(\ualpha^*)$ and extend it to the scapegoat qubit by
\begin{align}
\label{eq:X-trial}
  \Xlow \EqDef (1-\eps) X_\mathrm{max}\otimes A_1
    + \eps X_\sym(\ualpha^*)\otimes A_0 .
\end{align}

By the scapegoat definition, the energy contribution of the
$X_\mathrm{max}\otimes A_1$ is $1$, while the contribution of 
$X_\sym(\ualpha^*)\otimes A_0$ is zero. Therefore, the energy of
$\Xlow$ is at most $1-\eps$, creating a promise gap.

The above construction is a step in the right direction, but it does
not solve the problem completely. The reason is that not \emph{all}
$3$-local marginals of $\Xlow$ are PSD, and therefore it is not a
$3$-local qq state. While $\Xlow^{(ijk)}\succeq 0$ for every $i,j,k$
in the bulk, the problem is with marginals where $k=n+1$ is the
scapegoat, and $i,j$ are in the bulk. An intuitive way to see this
is to imagine that we ``measure'' $\Xlow$ at the $i,j,k$ qubits. The
$3$-local positivity condition means that this is a well-defined
procedure, and that we shall always have a positive probability for
each possible outcome. In other words, the overall negativity of the
operator should be \emph{hidden} from us, as long as the measurement
is $3$-local. But then assume that we first measure the scapegoat
qubit in the POVM basis $\{E_0, E_1, E_2, E_3\}$, and obtain the $0$
result. In such case, the collapsed state would be the branch that
is tensored with $A_0=D_0$, i.e., it would be $X_\sym(\ualpha^*)$,
which is not 2-local positive. This negativity might be revealed in
the two subsequent measurements, thereby violating the $3$-local
positivity condition.

We overcome this problem by using $3$ scapegoats instead of one. As
we shall see in the next subsection, we can start our reduction from
a graph $G=(V,E)$ that is decomposable into $3$ subgraphs, $G_1,
G_2, G_3$ with non-overlapping edges. With each subgraph, we
associate a scapegoat qubit, so that in total our Hamiltonian is
defined on $n+3$ qubits. The idea is that on each subgraph the
4-coloring problem has many solutions, and so we can find a
$2$-local, color invariant qq state of legal colorings of that
subgraph with \emph{strictly positive} 2-local marginals. We call
these states $\lambda$-solutions when the minimal eigenvalue of
their 2-local marginals is at least $\lambda >0$. We shall then use
a convex combination of the maximally mixed state, the three
$\lambda$-solutions of every scapegoat, and the $X_\sym(\ualpha^*)$
operator. Unlike our previous construction, now the negativity of
$X_\sym(\ualpha^*)$ is well-hidden: if we measure one of the
scapegoats, the collapsed state will contain a combination of the
$\lambda$-solution and $X_\sym(\ualpha^*)$, which will always be
$2$-local positive. The full details of our construction are given
in the next subsection.

\subsection{\texorpdfstring{$\lambda$}{lambda}-solutions and 
  \texorpdfstring{$3$}{3}-decomposable graphs}

A central ingredient in our solution is a $\lambda$-solution, which
we now define.
\begin{definition}[$\lambda$-solution]
\label{def:lam-sol} 
  Given a graph $G=(V,E)$ and a constant $\lambda \in (0,1/4)$, we
  say that $X=\sum_\ualpha \mu(\ualpha) D_\ualpha$ is a
  $\lambda$-solution of $G$ if 
  \begin{enumerate}
    \item $X$ is a $2$-local qq state,
    
    \item Each $\ualpha$ in the support of $X$ is a legal $4$-coloring
      of $G$,
      
    \item $X$ is color invariant (see \Def{def:color-inv}),
      
    \item For every $\{i,j\}\notin E$, we have  
        $X^{(ij)} \succeq \lambda\Id$.
    \label{cond:pos}
  \end{enumerate}
\end{definition}
We add two important remarks for the definition.
\begin{remark}
  As $X$ is color invariant, according to \Cor{corol:lambda-min},
  property (4) is equivalent to requiring that $\lambda \le
  \tau_{\mu}^{(ij)} \le \frac{1}{3}(1-\lambda)$, whenever
  $(i,j)\notin E$, where we recall that $\tau_{\mu}^{(ij)} =
  \Prob_\mu(\alpha_i = \alpha_j)$.
\end{remark}
\begin{remark}
\label{rmk:product}
  Let $G$ have connected components $G_1, \ldots , G_k$, and let
  $X=X_1\otimes \cdots \otimes X_k$ be a corresponding product
  operator.  If $X_a$ is a $\lambda$-solution for $G_a$, for $1 \leq
  a \leq k$, then $X$ is a $\lambda$-solution for $G$. Indeed,
  properties (1)--(3) follow from definition.  Property (4) follows
  from the color invariance of the individual components: if $i,j$
  are in different connected components of $G$ then
  $X^{(i)}=\frac{1}{2}\Id, X^{(j)}=\frac{1}{2}\Id$ and consequently,
  $X^{(ij)} = \frac{1}{4}\Id$, which satisfies property (4). 
\end{remark}

Next, we define the notion of a $(3,\lambda)$-decomposable graph,
with which we define our $3$-local Hamiltonian and the $3$
scapegoats.
\begin{definition}[$(3,\lambda)$-decomposable graph]
\label{def:3decomp-G} 
  Given a constant $\lambda\in (0,1/4)$, a 
  $(3,\lambda)$-decomposable graph is a graph $G=(V,E)$, which is
  given together with a decomposition of its edges into $3$
  non-intersecting subsets $E=E_1 \cup E_2 \cup E_3$ such that each
  of the $3$ subgraphs $G_{\ell} = (V, E_{\ell})$, $\ell = 1,2,3$,
  has a $\lambda$-solution. 
\end{definition}

Finally, we shall use the classical PCP theorem to prove that there
exist a family of $(3,\lambda)$-decomposable graphs with a large
promise gap which are NP-hard to $4$-color. This fact is expressed
in the following theorem, whose proof is given in
\Sec{sec:proof-of-3decomp}.
\begin{theorem}
\label{thm:3decomp-hardness} 
  There exist constants $\xi\in (0,1)$ and $\lambda\in(0,1/4)$ such
  that given a $(3,\lambda)$-decomposable graph $G$ that is either
  4-colorable or at least a fraction $\xi$ of its edges cannot be
  legally colored, it is {\rm NP}-hard to decide which of these
  two cases holds.
\end{theorem}

\subsection{Constructing the 3-local scapegoat Hamiltonian}

We are now ready to define our Hamiltonian. Given a graph
3-decomposable graph $G=(V,E)$ and the constant $\xi\in(0,1)$ from
\Thm{thm:3decomp-hardness}, we define a $3$-local Hamiltonian over
$n+3$ qubits: $n$ \emph{bulk qubits} that are associated with the
$n=|V|$ vertices of $G$ and $3$ scapegoat qubits that are labeled by
$n+1, n+2, n+3$.  Recalling the definition of the edge coloring
operator $h_e$ from \Eq{eq:coloring-h}, we define our Hamiltonian by
\begin{align}
  \label{def:H}
  H_G \EqDef \overbrace{\sum_{\ell=1}^3 \Big(\sum_{e\in E_{\ell}} 
    h_e\Big)\otimes F_0^{(n + \ell)}}^{\text{coloring constraints}} 
    + \overbrace{\xi|E|\cdot (\Id-F_0^{(n + 1)}
      \otimes F_0^{(n + 2)}\otimes 
      F_0^{(n + 3)})}^{\text{scapegoat constraint}} .
\end{align}
An illustration of the terms of this Hamiltonian is given in
\Fig{fig:scapegoats}. We call the $h_e\otimes F_0^{(n+\ell)}$ terms
in the Hamiltonian the coloring constraints and the second term the
scapegoat constraint. Note that by definition $\Id-F_0\otimes
F_0\otimes F_0 = \sum_{(\alpha_1,\alpha_2, \alpha_3)\neq (0,0,0)}
F_{\alpha_1}\otimes F_{\alpha_2}\otimes F_{\alpha_3}$, and therefore
the scapegoat constraint penalizes every configuration in which the
scapegoats are different than $D_0\otimes D_0\otimes D_0 =
A_0\otimes A_0\otimes A_0$.
\begin{figure}
  \begin{center}
    \includegraphics[scale=1.0]{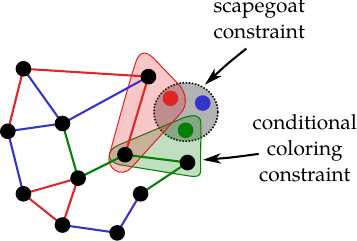}
  \end{center}
  \caption{Illustration of the construction of the $H_G$
    Hamiltonian, given in \eqref{def:H}. (a) We start with a
    3-decomposable graph, where the edges have been partitioned into
    $3$ subsets, red, greed, blue. (b) We place qubits on the
    vertices of the graph and add $3$ scapegoat qubits to the
    system, one for each subset. The terms in the local Hamiltonian
    are now $3$-local: there are coloring terms, which involve the
    qubits of an edge, together with its associated scapegoat, and
    there is the scapegoat constraint, which involves all three
    scapegoats.} \label{fig:scapegoats}
\end{figure}

To prove the quasi-quantum PCP theorem, \Thm{thm:main-NP-hard}, we
will prove the following theorem.
\begin{theorem}  
\label{thm:H-G-reduction} There exist constants $0<\xi_a<\xi_b<1$
  such that for every $G$ from the family of
  $(3,\lambda)$-decomposable graphs from \Thm{thm:3decomp-hardness},
  the corresponding Hamiltonian $H_G$ from \eqref{def:H} has the
  following properties:
  \begin{description}
    \item [Completeness:] If $G$ is $4$-colorable then the qq ground 
      energy of $H_G$ satisfies $\eps_0^{qq}\le \xi_a\cdot L_H$.
      
    \item [Soundness:] If $G$ is not $4$-colorable, then
      $\eps_0^{qq}\ge \xi_b\cdot L_H$.
  \end{description}
  Above,  $L_H$ is the approximation
  scale of $H_G$, given in \eqref{def:L-H-formula}.
\end{theorem}
Combining \Thm{thm:H-G-reduction} with \Thm{thm:3decomp-hardness}
proves the quasi-quantum PCP theorem~\ref{thm:main-NP-hard}.

In the next subsections we show that this reduction is sound and
complete. Note that, in an opposite manner to Kitaev's QMA-hardness
proof of the local-Hamiltonian problem\cc{ref:Kitaev2002-QI}, here
soundness is trivial, while completeness is quite involved.

\subsection{Proof of \Thm{thm:H-G-reduction}: soundness of the 
  reduction}

Assume that $G$ is not $4$-colorable, and let $\ualpha \in
\{0,1,2,3\}^n, \ubeta \in \{0,1,2,3\}^3$ be configurations of bulk
and scapegoat qubits respectively. Consider the operator
$D_\ualpha\otimes D_\ubeta$. If $\ubeta\neq(0,0,0)$, then due to the
scapegoat constraint in definition~\eqref{def:H} of $H_G$, we obtain
$\Tr(D_\ualpha\otimes D_\ubeta H_G) \ge \xi |E|$. On the other hand,
when $\ubeta = (0,0,0)$, all the coloring constraints in $H_G$ are
active. By assumption, at least $\xi|E|$ edges of $G$ are not
legally colored by $\ualpha$, and therefore also here the energy
will be at least $\xi|E|$. All in all, $\Tr(D_\ualpha\otimes
D_\ubeta H_G) \ge \xi |E|$ for every $D_\ualpha\otimes D_\ubeta$ and
so $\eps_0^{(qq)} \ge \xi|E|$.

Next, we use the explicit expressions $h_e =
\frac{1}{3}(\Id-\ketbra{\psi_-}{\psi_-})$, and $F_0 =
\frac{1}{2}\ketbra{0}{0}$, to conclude that
\begin{align*}
  L_H = |E|\cdot \frac{1}{3}\cdot\frac{1}{2} + \xi|E| = 
    |E|\cdot(\xi + \frac{1}{6}) .
\end{align*}
Setting $\xi_b \EqDef \frac{\xi}{\xi+
1/6}$, we get $\eps_0^{(qq)} \ge \xi_b L_H$.

\subsection{Proof of \Thm{thm:H-G-reduction}: completeness of 
  the reduction}

To show completeness, we will show that there exist constants $\eps,
\delta \in (0,1)$ such that when the 3-decomposable $G$ is
4-colorable, there exists a 3-local qq state $\Xlow$ with $\Tr(\Xlow
H_G)\le \xi|E|(1-\eps\cdot\delta)$. Taking $\xi_a =
\frac{\xi}{\xi+1/6} (1-\eps\cdot\delta) = \xi_b(1-\eps\cdot\delta)$
then proves completeness. 

We start by defining the low energy state. It is essentially a
generalization of the state defined in \Sec{sec:bird-eye} to the
case of $3$ scapegoat qubits.  
\begin{definition}[The low energy state]
  Let $\delta, \eps > 0$ be constants to be determined later. Let
  $\ualpha^*$ be a legal coloring of the $(3, \lambda)$-decomposable
  graph $G$, and let $X_1 X_2, X_3$ be $\lambda$-solutions
  respectively for the subgraphs $G_1, G_2, G_3$. Then we define
  \begin{align}
  \label{def:Xlow}
  \begin{split}
    \Xlow &\EqDef (1-\eps)X_\mathrm{max}\otimes A_1A_1A_1\\
       &+ \eps\Big[\delta X_0\otimes A_0A_0A_0 +
      \frac{1-\delta}{3}(X_1\otimes A_0A_1A_1 
        + X_2\otimes A_1A_0A_1 
        + X_3\otimes A_1A_1A_0)\Big] ,
  \end{split}     
  \end{align}
  where $X_\mathrm{max} \EqDef \frac{1}{2^n}\Id$ is the
  maximally-mixed state, and $X_0=X_\sym(\ualpha^*)$ is the
  symmetric coloring operator of $\ualpha^*$. In addition, the
  notation $A_iA_jA_k$ for the state of the scapegoat qubits means
  $A_i\otimes A_j\otimes A_k$.
\end{definition}
To finish the completeness proof, we will show:
\begin{enumerate}
  \item $\Tr(\Xlow H_G) \le \xi|E|(1-\eps\delta)$.
  \item $\Xlow$ is a 3-local qq state.
\end{enumerate}

\subsubsection{Upper bounding 
  \texorpdfstring{$\Tr(\Xlow H_G)$}{Tr(Xlow HG)}}

By construction, all the parts in $\Xlow$ in which at least one
scapegoat is active (i.e., parts coming from $X_{max}, X_1, X_2,
X_3$), have no energy penalty from the coloring constraints, and get
a $\xi|E|$ energy from the scapegoats constraint. On the other hand,
the term $X_0\otimes A_0A_0A_0$ has zero energy, and therefore
\begin{align*}
  \Tr(\Xlow H_G) = \xi|E|(1-\eps\delta) .
\end{align*}

\subsubsection{\texorpdfstring{$\Xlow$}{Xlow} is a 
  \texorpdfstring{$3$}{3}-local qq state}

To show that $\Xlow$ given in \eqref{def:Xlow} is a $3$-local qq
state, we need to show that for every $3$ qubits $i,j,k$, the
marginal $\Xlow^{(ijk)}$ is a PSD. Recalling that we refer to the
qubits associated with the vertices of $V$ as bulk qubits, and to
$n+1,n+2,n+3$ as scapegoat qubits, we consider 4 cases:
\begin{enumerate}

  \item \textbf{$i,j,k$ are all in the bulk} 
  
    In such case $X^{(ijk)}_\mathrm{low}$ is given by
    \begin{align*}                                                     
      \Xlow^{(ijk)} = (1-\eps) \frac{1}{8}\Id + \eps \big[
        \delta X_{0}^{(ijk)} + \frac{1-\delta}{3}(
          X_{1}^{(ijk)} + X_{2}^{(ijk)} + X_{3}^{(ijk)})\big].
    \end{align*}
    Note that as $-\Id\preceq D_\alpha\preceq 2\Id$ for every
    $\alpha$ (see \Eq{eq:D2}), it follows for any $k$-qubits
    operator of the form $X=\sum_\ualpha \sigma(\ualpha)
    D_{\alpha_1}\otimes \ldots \otimes D_{\alpha_k}$, where
    $\sigma(\ualpha)$ is a distribution, we have $-2^{k-1}\Id
    \preceq X\preceq 2^k \Id$, and $\norm{X}\le 2^k$. Therefore,
    $\norm{X_\ell^{(ijk)}}\le 8$ for $\ell=0,1,2,3$, and so it is
    clear that for sufficiently small $\eps$, $\Xlow^{(ijk)}$ is
    PSD.
    
  \item \textbf{$i,j$ are in the bulk, $k$ is a scapegoat}
  
    Assume without loss of generality that $k=n+1$. Following the
    discussion in \Sec{sec:bird-eye}, now we should also take into
    account the effect of the first scapegoat qubit on the
    positiveness of $\Xlow^{(ijk)}$. The scapegoat can be in two
    states: either $A_0=D_0=2\ketbra{0}{0}-\ketbra{1}{1}$ or
    $A_1=\frac{1}{3}(D_1+D_2+D_3) = \ketbra{1}{1}$. We will show
    that: (i) the (operator) coefficient of $A_1$ in the expansion
    of $\Xlow^{(ijk)}$ is a PSD, and (ii) the the coefficient in
    front of $A_0$ is also PSD \emph{and} is much smaller than the
    first coefficient. This would guarantee that the
    $-\ketbra{1}{1}$ element in $A_0$ will be compensated by the
    positiveness of the $A_1=\ketbra{1}{1}$ coefficient. 
            
    Formally, from \Eq{def:Xlow}, we get
    \begin{align}       
    \label{eq:case-ii}
      \Xlow^{(ijk)} &= (1-\eps) \frac{1}{4}\Id\otimes A_1
        + \eps\delta X_0^{(ij)}\otimes A_0 + 
          \eps(1-\delta)\frac{1}{3} X_1^{(ij)}\otimes A_0\\
        &+ \eps(1-\delta)\frac{1}{3} 
            \big( X_2^{(ij)}+ X_3^{(ij)}\big)\otimes A_1 .
         \nonumber
    \end{align}
    Gathering the $A_0,A_1$ contributions yields
    \begin{align*}       
      \Xlow^{(ijk)} &= K_1\otimes A_1 + K_0\otimes A_0
        = (K_1 -K_0)\otimes \ketbra{1}{1} + 2K_0\ketbra{0}{0}
    \end{align*}
    where $K_0 \EqDef \eps\Big(\delta X_0^{(ij)} +
    (1-\delta)\frac{1}{3} X_1^{(ij)}\Big)$ and $K_1 \EqDef
    (1-\eps)\frac{1}{4}\Id + \eps(1-\delta)\frac{1}{3}\big(
    X_2^{(ij)} + X_3^{(ij)}\big)$. We are then left to show that
    $K_1-K_0\succeq 0, K_0\succeq 0$.
    
    To show that $K_0\succeq 0$, first assume that $\{i,j\}$ is
    \emph{not} an edge in $G_1$. By assumption $X_1$ is a
    $\lambda$-solution of $G_1$ and therefore in such case,
    $X_1^{(ij)} \succeq \lambda\Id$ (see \Def{def:lam-sol}). Taking
    $\delta=\lambda/10$ and using the fact that the smallest
    eigenvalue of $X_0^{(ij)}$ is $-2$, we find that the minimal
    eigenvalue of $K_0$ is at least $\frac{1}{200}\lambda$.  Next,
    assume that $\{i,j\}$ is an edge in $G_1$. By
    \Lem{lem:2marginals}, in such case both $X_0^{(ij)}$ and
    $X_1^{(ij)}$ are equal to $A = \ketbra{\psi_-}{\psi_-}\succeq 0$
    (see \Eq{def:SA}), and therefore also $K_0\propto A\succeq 0$.
    
    Finally, the positiveness of $K_1-K_0$ follows immediately from
    the $(1-\eps)\frac{1}{4}\Id$ term with $\eps=1/40$.

  \item \textbf{$i$ is in the bulk while $j,k$ are scapegoats}
  
    In such case, from the color invariance of $\Xlow$, it follows
    that the marginal on the single qubit $i$ in the bulk is equal
    to $\frac{1}{2}\Id$ and therefore $\Xlow^{(ijk)} =
    \frac{1}{2}\Id\otimes \Xlow^{(jk)}$. Consequently, we only need to
    show that $\Xlow^{(jk)}\succeq 0$. Without loss of generality, take
    $j,k=1,2$. Following the outline of the previous case,
    $\Xlow^{(jk)}$ is given by
    \begin{align*}
      \Xlow^{(jk)} &= (1-\eps)A_1\otimes A_1 + \eps \Big( 
      \delta A_0\otimes A_0 + \frac{1-\delta}{3} \big(A_0\otimes A_1
        + A_1\otimes A_0 + A_1\otimes A_1\big) \Big)\\ 
     &\EqDef K_{00}A_0\otimes A_0 + K_{01}A_0\otimes A_1
      + K_{10}A_1\otimes A_0 + K_{11}A_1\otimes A_1,
    \end{align*}
    To show that $\Xlow^{(ij)}$ is PSD, we follow the similar
    arguments as in the previous case. We rewrite the $A_0,A_1$ in
    the computational basis, and verify that we get non-negative
    coefficients. For example, the coefficient of $\ketbra{11}{11}$
    will be
    \begin{align*}
      K_{11} - K_{10} - K_{01} + K_{00}.
    \end{align*}
    As $K_{11} = 1-\eps$ while the rest of the terms are either
    proportional to $\eps$, we deduce that the overall coefficient
    is positive. The rest of the coefficients follow a similar
    behavior.

  \item \textbf{$i,j,k$ are scapegoats}
  
    This case follows the same reasoning as the previous cases. By
    definition, 
    \begin{align*}
      \Xlow^{(ijk)} &= (1-\eps)A_1\otimes A_1\otimes A_1 + \eps \Big( 
        \delta A_0\otimes A_0\otimes A_0 \\
        &+ \frac{1-\delta}{3}
        \big(A_0\otimes A_1\otimes A_1 
          + A_1\otimes A_0\otimes A_1 
          + A_1\otimes A_1\otimes A_0\big) 
          \Big).
    \end{align*}
    As in the previous case, positivity can proven by rewriting
    $\Xlow^{(ijk)}$ in the computational basis and showing the
    positivity of each coefficient.
\end{enumerate}

This concludes the soundness proof.

\subsection{Proof of \Thm{thm:3decomp-hardness}}
\label{sec:proof-of-3decomp}

The theorem follows directly from the next two
lemmas:
\begin{lemma}
\label{lem:max-deg-4}
  There is a constant $\xi\in (0,1)$ such that given a graph
  $G=(V,E)$ with a maximal degree $4$, which is either $3$-colorable,
  or at least a fraction of $\xi$ of its edges cannot be
  $3$-colored, it is {\rm NP}-hard to decide between these two cases.
\end{lemma}

\begin{lemma}
\label{lem:3-to-4}
  There is a constant $\lambda\in(0,1/4)$ and an efficient map that
  takes any graph $G'=(V',E')$ with a maximal degree $4$ to a new
  graph $G=(V,E)$ such that:
  \begin{itemize}  
    \item $G$ is $(3,\lambda)$-decomposable 
      
    \item If $G'$ is $3$-colorable then $G$ is $4$-colorable.
    
    \item If at least a fraction $\xi'$ of the edges in $G'$ cannot
      be $3$-colored, then at least a fraction $\xi$ of the edges in
      $G$ cannot be $4$-colored, where $\xi=\Omega(\xi')$.
  \end{itemize}
\end{lemma}

\begin{proof}[\ of \Lem{lem:max-deg-4}]
 
  Our starting point is Theorem~1.1 from
  \cc{ref:Austrin2014-3colorPCP}, which uses the PCP theorem to show
  that given a graph $G'$ which is either $3$-colorable or at least
  a $\xi' = \frac{1}{18}$ fraction of its edges cannot be legally
  $3$-colored, it is NP-hard to decide between these cases.
  
  To prove the lemma, we need to reduce the degree of each vertex to
  be at most $4$. Following the degree reduction procedure in
  \cite{ref:Dinur2007-PCP}, we replace every vertex in $G'$ of degree
  $D>4$ with a ``cloud'' of $D$ vertices that sit on a degree $3$
  expander graph --- see \Fig{fig:local-expander}. We define
  equality constraints between the edges of that expander and denote
  the resultant constraint graph by $G''$. By Lemma~4.1 of
  \cite{ref:Dinur2007-PCP}, if $G'$ is $3$-colorable, then
  $G''$ is satisfiable. On the other hand, if a fraction $\xi'$
  of edges of $G'$ could not be colored, then a fraction
  $\xi''=\Omega(\xi')$ of the edges in $G''$ cannot be
  satisfied.

  \begin{figure}
    \begin{center}
      \includegraphics[scale=0.2]{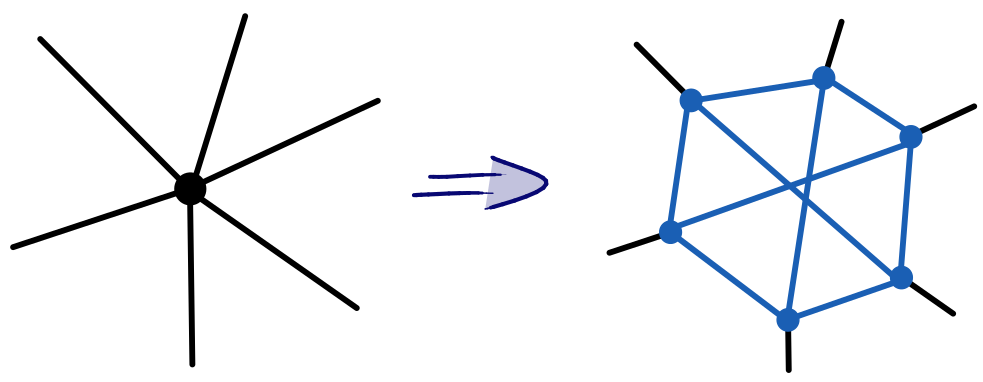}
    \end{center}
    \caption{Reducing the local degree of a constraint graph. Given
      a vertex with degree $D$, we replace it with $D$ vertices that
      sit on an a degree $3$ expander graph. The edges of the
      expander graph (colored in blue) are associated with equality
      constraints.} \label{fig:local-expander}
  \end{figure}

  \begin{figure}
    \begin{center}
      \includegraphics[scale=0.2]{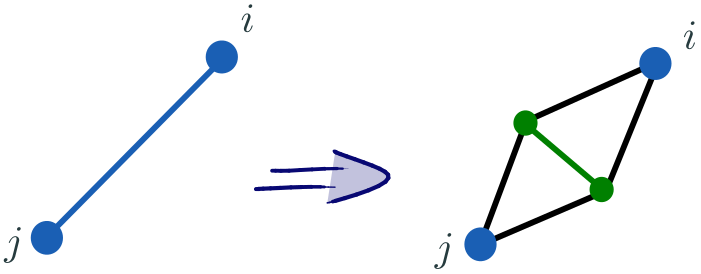}
    \end{center}
    \caption{A $3$-coloring gadget for an equality constraint
    between vertices $i,j$}
    \label{fig:3-color-ID}
  \end{figure}

  Next, we turn the constraint graph $G''$ back to a $3$-coloring
  problem of a new graph $G'''$. We do that by locally replacing
  each equality constraint by the $3$-coloring gadget that is
  depicted in \Fig{fig:3-color-ID}. It is clear that if $G''$ is
  satisfiable then $G'''$ is $3$-colorable, and if at least a
  fraction ${\xi''}$ edges of $G''$ cannot be satisfied, then a
  fraction $\xi''' \ge \xi''/5$ of the edges in $G'''$ cannot be
  legally 3-colored.

  \begin{figure}
    \begin{center}
      \includegraphics[scale=0.2]{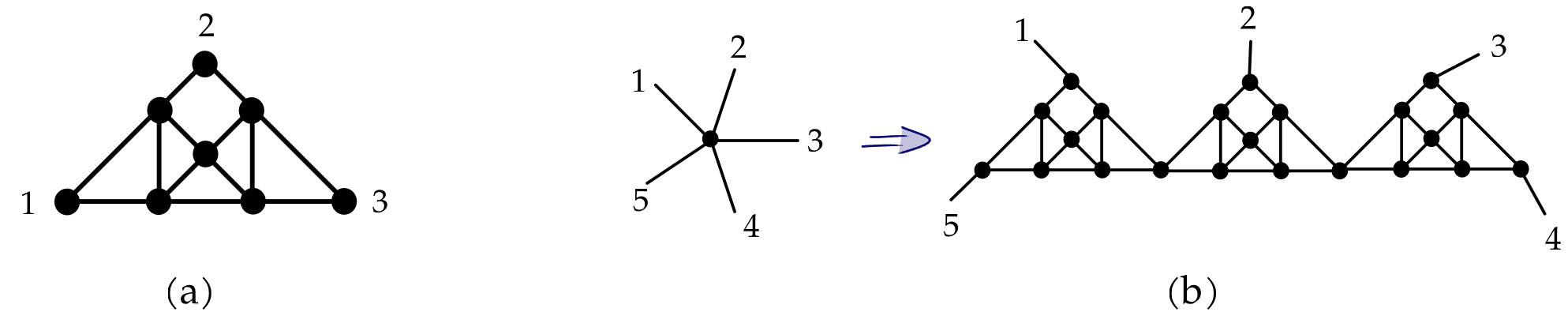}
    \end{center}
    \caption{Reducing the local degree of a $3$-coloring graph
    problem to $4$. (a) The basic gadget, which forces vertices
    $1,2,3$ to be colored identically. (b) Applying the basic gadget
    several times allows us to reduce any degree to $4$.}
    \label{fig:deg4-reduction}
  \end{figure}  

  At this point, the local degree of $G'''$ is at most $7$: there
  are $3\times 2 = 6$ edges come from the identity constraints of
  the degree $3$ expander, and they are connected to a single
  external edge. To finish the proof, we reduce the degree further
  to $4$ and call the resulting graph $G=(V,E)$. To this aim, we use
  the coloring gadgets of \cite{ref:Garey1976-NP}, illustrated in
  \Fig{fig:deg4-reduction}(a), which forces vertices $1,2,3$ to be
  colored identically. Therefore, if we wish, for example, to reduce
  a degree $5$ vertex to a degree $4$, we apply this gadget
  iteratively $3$ times, as illustrated in
  \Fig{fig:deg4-reduction}(b). As in the previous steps, also this
  map reduces the promise gap by at most a constant factor, as each
  vertex with degree $5,6,7$ is replaced by a gadget with a $O(1)$
  edges. This concludes the proof.
\end{proof}

\begin{proof}[\ of \Lem{lem:3-to-4}]
  Given a graph $G'=(V',E')$ of a $3$-coloring problem, we create
  the graph $G=(V,E)$ in three steps that are illustrated in
  \Fig{fig:kafka}. Below, we describe these steps, and show that the
  promise gap is preserved in all of these steps up to a constant
  factor.

  \begin{figure}
    \begin{center}
      \includegraphics[scale=0.2]{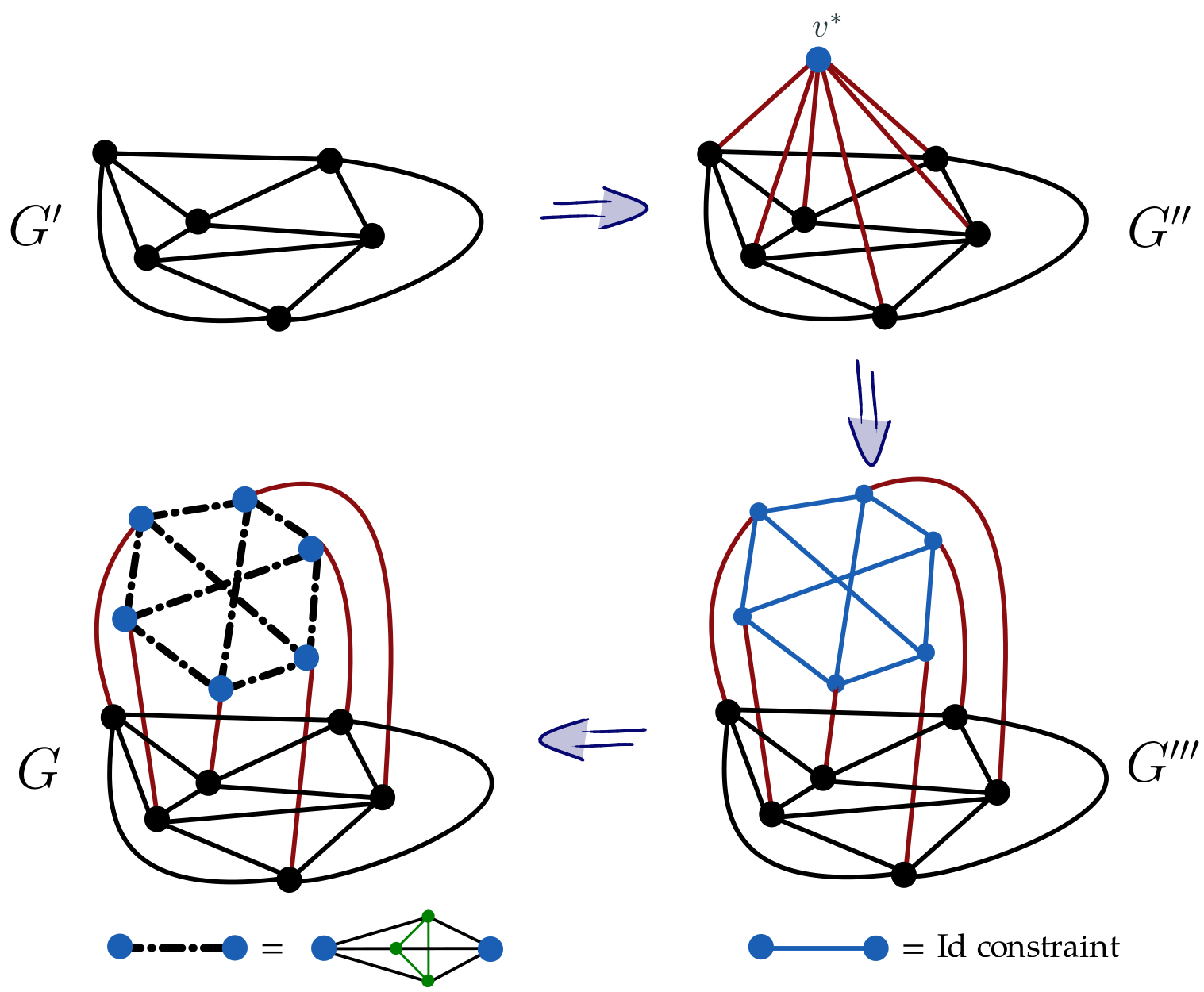}
    \end{center}
    \caption{Illustration of the $3$ reduction steps in the proof of
    \Lem{lem:3-to-4}: \textbf{Step I:} The $3$-coloring problem of
    $G'$ is turned to a $4$-coloring problem of $G''$ by adding a
    single vertex $v^*$ and connecting it to all the vertices of
    $G'$. \textbf{Step II:} to reduce the degree of $v^*$, it is
    replaced by a ``cloud'' of vertices $W$, which are all forced to
    have the same color using identity constraints (blue edges) that
    are placed on the edges of a degree-$3$ expander. The brown
    edges connecting the vertices of $W$ to those of $G'$ are called
    \emph{cross} edges.  \textbf{Step III:} The identity constraints
    among the vertices $W$ are replaced by a local gadget, described
    in \Fig{fig:4-color-ID}. This turns every equality constraint in
    $G'''$ into $9$ edges, defined the original two vertices and
    additional $3$ vertices (painted in green). The set of all these
    additional vertices is denoted by $T$.}  \label{fig:kafka}
  \end{figure}

  \paragraph{Step I:} 
  In the first step, we pass to a 4-coloring problem by adding a new
  vertex $v^*$ to the system and connecting all $v'\in V'$ to that
  vertex. We denote this new graph by ${G}'' = ({V}'', {E}'')$.
  Clearly, the addition of $v^*$ forces all the original vertices to
  be colored in one of the $3$ colors different than the color of
  $v^*$, implying that ${G}''$ is $4$-colorable if $G'$ is
  3-colorable. Assume that at least a fraction $\xi'$ of the
  edges of $G'$ is not $3$-colorable. The next claim shows that in
  such at least a fraction ${\xi}''\ge \xi'/5$ of the edges of
  ${G}''$ cannot be legally $4$-colored.  
  \begin{claim}
    If at least a fraction $\xi'$ of the edges in $G'$ cannot be
    legally $3$-colored, then at least a fraction of ${\xi}'' \ge
    \xi'/8$ edges in $G''$ cannot be legally $4$-colored.
  \end{claim}
    
  \begin{proof}
    Let $(\alpha_1, \ldots, \alpha_n, \alpha^*)$ be an
    optimal $4$-coloring of $G''$, and assume without loss of
    generality that $\alpha^*=0$. The number of violated edges in
    $G''$ is then $\xi''|E''|$. Let $U\EqDef \{i\, :\, \alpha_i =
    0\}$ denote the subset of vertices of $G'$ that are colored with
    the color $0$. Clearly, $\xi''|E''| \ge |U|$. 
    
    We create a $3$-coloring of $G'$ by changing the color of every
    vertex in $U$ to one of the remaining colors $1,2,3$. Let us now
    consider this $3$-coloring of $G'$ as a $4$-coloring of $G''$ by
    keeping the color of $v^*$ be $0$.  Let $q$ be the number of
    violations in $G''$ of this new $4$-coloring. On one hand, since
    all the edges adjacent to $v^*$ are properly colored, then all
    violations come from violations of the $3$-coloring of $G'$.
    Therefore, $q\ge \xi'|E'|$. On the other hand, by flipping the
    color of every $v'\in U$, we have introduced at most $4$ new
    violations with respect to the optimal $4$-coloring we
    started with (because $v'$ is connected to a most $4$ other
    vertices in $V'$) and ``cured'' at least one violation (of the
    edge connecting $v'$ to $v^*$). Therefore, $q\le \xi''|E''| +
    (4-1)|U|$, and we conclude that
    \begin{align*}
      \xi''|E''| + 3|U| \ge \xi'|E'| .
    \end{align*}
    Adding to it, the condition $\xi''|E''| \ge |U|$, we get
    \begin{align*}
      \xi''|E''| \ge \max\big(\xi'|E'|- 3|U|, |U|\big).
    \end{align*}
    The minimum of the above RHS is obtained when $|U| =
    \frac{1}{4}\xi'|E'|$, and therefore $\xi''|E''| \ge
    \frac{1}{4}\xi'|E'|$, from which we get
    \begin{align*}
      \xi'' \ge \frac{1}{4}\cdot\frac{|E'|}{{|E''|}}\cdot \xi' .
    \end{align*}
    By definition, $|E''| = |E'| + |V'|$.  Without loss of
    generality, we assume that $|V'|\le |E'|$ and so, $|E''|\le
    2|E'|$, which implies
    \begin{align*}
      \xi'' \ge \frac{1}{4}\cdot \frac{1}{2} \cdot \xi' 
        = \frac{\xi'}{8} .
    \end{align*}
  \end{proof}
  
  \paragraph{Step II:} In the second step, we reduce the degree of
  $v^*$ using the techniques of \cite{ref:Dinur2007-PCP}. We replace
  the vertex $v^*$ by a cloud of vertices $W$ of the same number 
  as $|V'|$. We connect every vertex $w\in W$ to exactly one
  vertex $v\in V'$, and call these
  connections \emph{cross} edges. Then
  we place the vertices $W$ on a degree-3
  expander with equality constraints on its edges. We obtained a
  constraint graph over an alphabet of size $4$, which we denote by
  $G'''$. By Lemma 4.1 of \cite{ref:Dinur2007-PCP} we conclude
  that $G'''$ is satisfiable if and only if $G''$ is
  $4$-colorable, and that if at least a fraction $\xi''$ of the
  edges in $G''$ cannot be properly $4$-colored, then at least a
  fraction $\xi''' = \Omega(\xi'')$ of edges in $G'''$
  cannot be satisfied.

  \begin{figure}
    \begin{center}
      \includegraphics[scale=0.2]{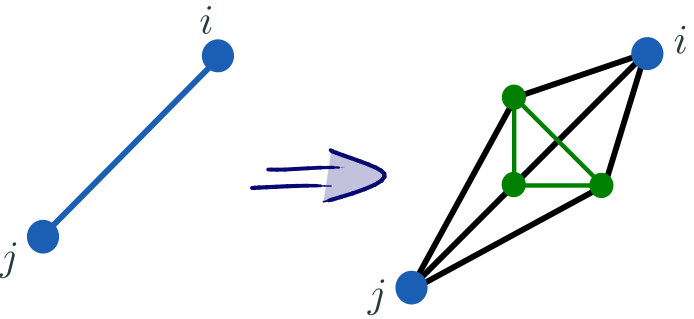}
    \end{center}
    \caption{A $4$-coloring gadget that enforces identical coloring
      of vertices $i,j$.}
    \label{fig:4-color-ID}
  \end{figure}

  \paragraph{Step III: } In the final step, we return to a
  $4$-coloring problem over the final graph $G=(V,E)$, which is
  defined by replacing every equality constraint in the expander
  graph with the $4$-coloring gadget of \Fig{fig:4-color-ID}.  This
  gadget replaces the edge of the equality constraint by $9$ new
  edges defined on $5$ vertices: the original two vertices of the
  equality constraint, plus $3$ new vertices (colored in green in
  \Fig{fig:4-color-ID} and \Fig{fig:kafka}). We denote the set of
  the newly introduced vertices by $T$, and note that as there are
  at most $3|V'|/2$ equality constraints in $G'''$, then $|T|\le
  9|V'|/2$.
  
  Clearly, if $G'''$ was satisfiable, then $G$ is $4$-colorable. On
  the other hand, as every equality constraint in $G'''$ was
  replaced by a gadget with $9$ edges, then $|E| \le 9|E'''|$.
  Therefore, if at least a fraction of $\xi'''$ of the edges in
  $G'''$ cannot be legally colored, then at least a fraction $\xi\ge
  \xi'''/9$ of edges in $G$ cannot be $4$-colored legally.

  Let us summarize the structure of the graph $G=(V,E)$ that we have
  obtained thus far (see the last graph in \Fig{fig:kafka}). It has
  3 types of vertices
  \begin{enumerate}
    \item The original vertices of $G'$, which are denoted by $V'$
      (black vertices in \Fig{fig:kafka}).
     
    \item The original vertices of the expander that were introduced 
      in step II and are denoted by $W$ (blue vertices in
      \Fig{fig:kafka}), where each $w\in W$ is connected to a unique
      $v\in V'$.

    \item The additional (green) vertices that were introduced in 
      step III by the equality gadget of \Fig{fig:4-color-ID}, and are
      denoted by $T$, where $|T| \le 9|V'|/2$.
    
  \end{enumerate}
  We can therefore view $G$ as two graphs: the original graph $G'$,
  and an expander-like graph $\Gamma$ that is defined on $W\cup T$,
  the vertices of type 2 and 3. Additionally, there are exactly
  $|V'|$ cross edges connecting the vertices of $V'$ to the vertices
  $W$ in $\Gamma$ (brown edges in \Fig{fig:kafka}).
  
  To show that $G$ is ($3$, $\lambda$)-decomposable, we would now
  like to partition $E$ into $3$ non-overlapping subsets $E=E_1\cup
  E_2\cup E_3$.  We will do that by first partitioning the edges of
  $G'$ into $3$ subsets and then partition the edges of $\Gamma$ to
  $3$ subsets.  We will then arbitrarily pair each of the subsets of
  $G'$ with a subset of $\Gamma$ and merge the two subsets in each
  pair. Finally, we will distribute the cross edges among the $3$
  merged subsets.
  
  To partition $G$, we use the following lemma, whose proof is given
  in \App{app:4decomp-proof}.
  \begin{lemma}
  \label{lem:4decomp} 
     Let $G'=(V',E')$ be a graph where the degree of every vertex is
     at most $4$. Then its set of edges $E'$ can be efficiently
     partitioned into $3$ disjoint subsets $E' = E'_1 \cup E'_2 \cup
     E'_3$ such that each of the subgraphs $G'_{\ell} = (V',
     E'_{\ell}),$ for $\ell =1,2,3,$ consists only of simple cycles,
     paths and isolated vertices.
  \end{lemma}
  
  To partition $\Gamma$, we partition every equality constraint into
  $3$ parts, as shown in \Fig{fig:expander-decomp}(a), and arbitrarily
  label these parts by $a,b,c$. We then merge all subsets with the
  same label together, resulting in $3$ disjoints subsets of the  
  edges in $\Gamma$, see \Fig{fig:expander-decomp}(b). 
  
  \begin{figure}
    \begin{center}
      \includegraphics[scale=0.2]{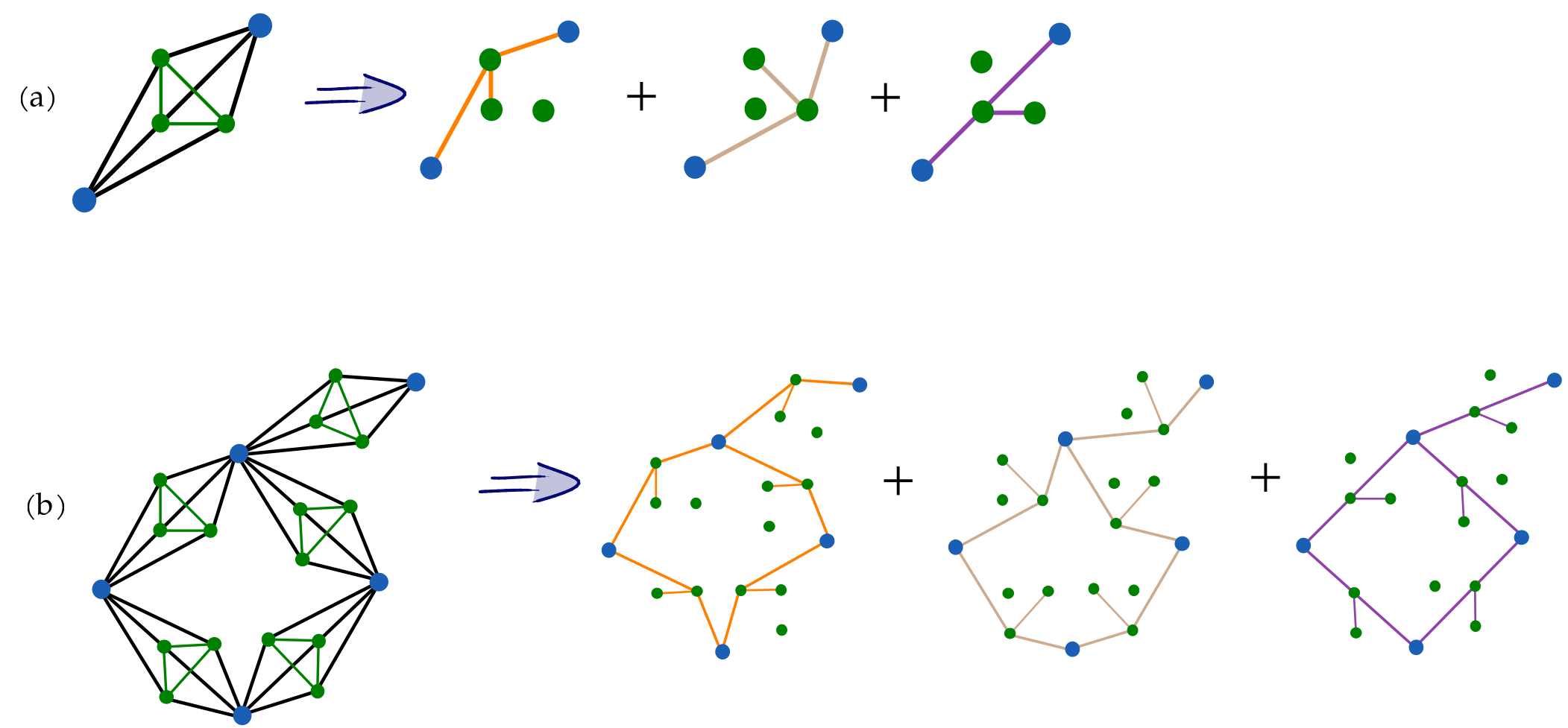}
    \end{center}
    \caption{ Decomposing
    the edges of $\Gamma$ to $3$ subsets. (a) every identity
    constraint gadget is decomposed into $3$ parts. (b) These parts
    are joined together arbitrarily, leading to a global
    decomposition of the edges in $\Gamma$.} 
    \label{fig:expander-decomp}
  \end{figure}
  
  Next, we merge the three edge subsets of $G'$ with those of
  $\Gamma$ by \emph{arbitrarily} joining an edge subset of $G'$ with
  an edge subset of $\Gamma$. This gives us $3$ disjoint subsets of
  the combined edges in $G'$ and $\Gamma$. We denote these subsets by
  $E_1, E_2, E_3$. 

  \begin{figure}
    \begin{center}
      \includegraphics[scale=0.2]{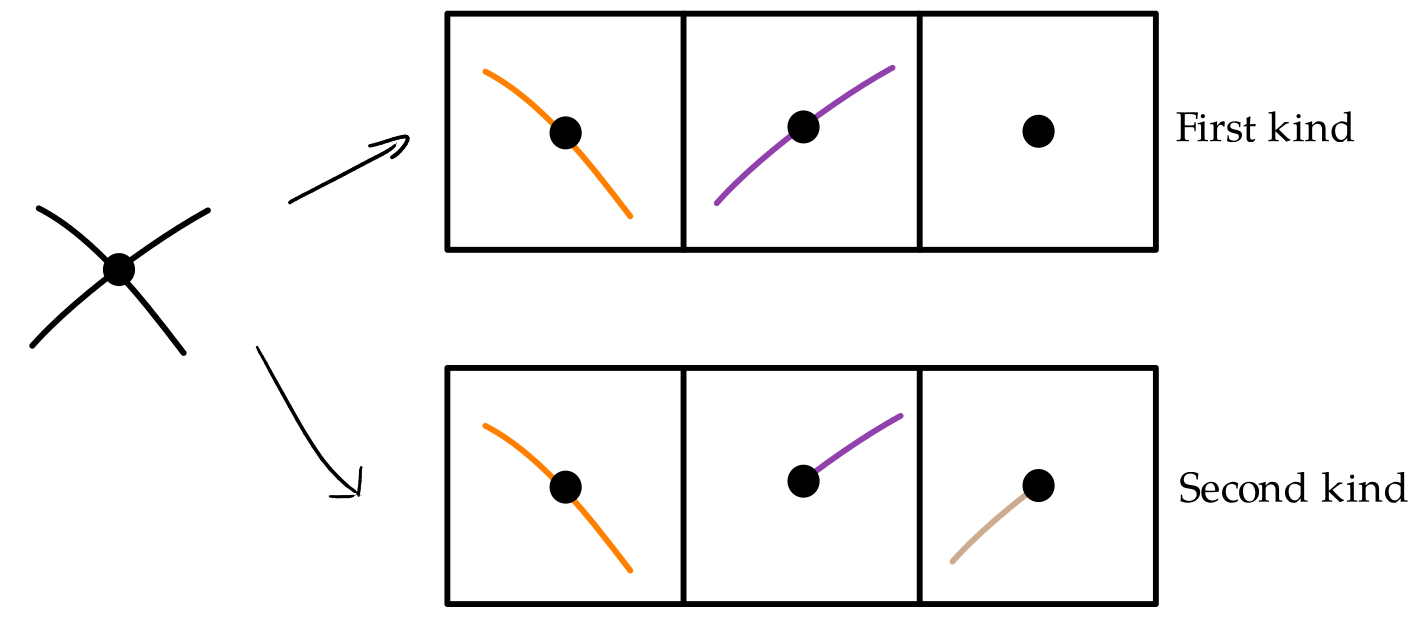}
    \end{center}
    \caption{The two kinds of vertices in $G'$, which are defined by
      the decomposition of $G'$ into $3$ subgraphs according to
      \Lem{lem:4decomp}.} 
  \label{fig:4decomp} 
  \end{figure}

  Finally, we need to associate each of the cross edges to one of
  the $E_1, E_2, E_3$ subsets. Our goal is to do that \emph{without
  creating new loops} in each of the subgraphs. For that, we use the
  following algorithm. Starting with the edge decomposition of $G'$,
  $E' = E'_1 \cup E'_2 \cup E'_3$, we are guaranteed that in the $3$
  induced subgraphs $G'_{\ell} = (V', E'_{\ell}),$ for $\ell
  =1,2,3,$ every vertex has degree at most 2. Note that with respect
  to the edge decomposition of \Lem{lem:4decomp}, every $v\in V'$
  can only be one of two kinds, which are illustrated in
  \Fig{fig:4decomp}. Then we go over all the vertices $v'\in V'$ and
  add the edge $e=(v', w)$, which connects $v'$ with the vertex
  $w\in W$ of $\Gamma$, to one of the subsets in $\{E_1, E_2, E_3\}$
  as follows. We first pick all the vertices of $v'\in V'$ of the
  first kind, and add $e = (v',w)$ to the subset where $v'$ is an
  isolated vertex.  Obviously, this can not create a loop.
  
  We are left with adding the cross edges $e=(v', w)$ that are
  adjacent to a vertex $v' \in V'$ of the second kind. We will add
  $e'$ to one of the two sets in $\{E_1, E_2, E_3\}$ where $v'$ 
  has degree one. Observe that having degree one means that $v'$ is
  one of the two leaves of a path. We avoid creating a loop if the
  cross edges corresponding to the two endpoints of a path are never
  added to the same set. To ensure that, while there is a second
  kind vertex $v'$ such that $e=(v', w)$ was not yet added to any of
  the sets in $\{E_1, E_2, E_3\}$, we do the following. We start by
  picking one of these vertices $v'$ and arbitrarily choose an edge
  set $E_\ell \in \{E_1, E_2, E_3\}$ where $v'$ has degree 1. We add
  $e$ to $E_\ell$.  Let $v''$ be the other leaf of the path where
  $v'$ belongs in $E_\ell$, it is also a vertex of the second kind.
  Necessarily, there is a second subset $E_{\ell'}\neq E_\ell$ in
  which $v''$ is a leaf in a path. We then add the cross edge of
  $v''$ to $E_{\ell'}$ and proceed as before and conclude this part
  when all the cross edges were added.
    
  To finish the proof, we need to show that for each of the
  subgraphs that we created there exists a $\lambda$-solution. We
  note that by construction, every subgraph is a disjoint union of
  graphs of the following form:
  \begin{enumerate}
    \item Isolated vertices
    \item Isolated cycles
    \item Graphs of the third kind, which are defined below in
      \Def{def:third-kind}. An illustration of such graph is given in
      \Fig{fig:third-kind}
  \end{enumerate}
  
  \begin{definition}[Graphs of the third kind]
  \label{def:third-kind} A graph of the third kind is a graph that
    can be created as follows. Start from a graph $J$ of a maximal
    degree $4$ and place a vertex in the middle of each of its
    edges. Then go over all the mid vertices, and for each such
    vertex $v$ optionally add a new vertex $v'$ and connect the two
    by an edge.
  \end{definition}
  
  \begin{figure}
    \begin{center}
      \includegraphics[scale=0.2]{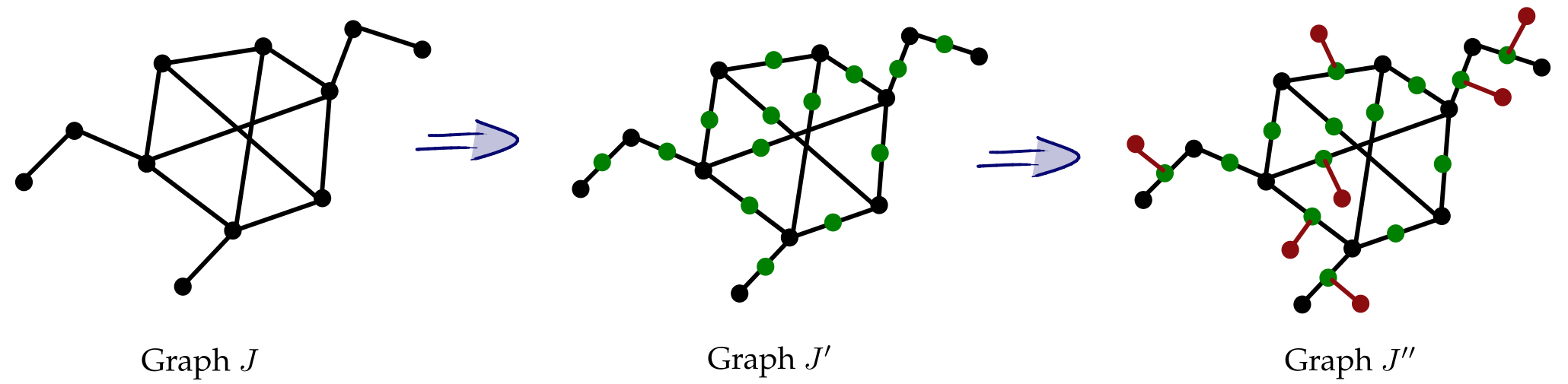}
    \end{center}
    \caption{Illustrating the definition of a graph of the third
      kind (\Def{def:third-kind}). Starting with a graph $J$ of a
      maximal degree, we create a graph $J'$ by placing a new vertex
      in the middle of every edge of $J$. Then we create $J''$ by
      going over the mid vertices and optionally connecting them to
      new vertices. } \label{fig:third-kind}
  \end{figure}
  
  By Remark~\ref{rmk:product}, to show that our subgraphs have a
  $\lambda$-solution, it is sufficient to show that cycles and
  graphs of the third kind have $\lambda$-solutions for some fixed
  $\lambda\in (0,1)$. (for isolated vertices we can take
  $X=\frac{1}{2}\Id$). This is shown in the following two lemmas,
  whose proofs are given in \App{sec:lambda-min}.
  \begin{lemma}
  \label{lem:cycle}
    Every cycle has a $\lambda$-solution for a universal $\lambda\in
    (0,1)$.
  \end{lemma}

  \begin{lemma}
  \label{lem:third-kind} 
    Every graph of the third kind has a
    $\lambda$-solution for a universal $\lambda\in (0,1)$.
  \end{lemma}

\end{proof}

\section{Acknowledgments}

We thank Dima Gavinsky, Thomas Vidick, Chen Zhili and Anurag Anshu
for valuable discussions.

This project is supported by the National Research Foundation,
Singapore through the National Quantum Office, hosted in A*STAR,
under its Centre for Quantum Technologies Funding Initiative
(S24Q2d0009)

{~}

\begin{center}
   {\Large \textbf{APPENDICES}}
\end{center}


\appendix

\section{Proof of \Lem{lem:q-support}}
\label{app:q-support}

Assume that $X$ is a PSD and write $X = \sum_\ualpha \mu(\ualpha)
D_\ualpha$. Let $\ualpha\in \supp(X)$. For every $i=1, \ldots, n$,
let $S^{(i)}_\ualpha\subseteq \supp(X)$ be the set of strings in the
support of $X$ that coincide with $\ualpha$ at all locations except
for the site $i$. We claim that necessarily $|S^{(i)}_\ualpha|\ge
3$.  To see this, define 
\begin{align*}
  R \EqDef \bigotimes_{j\neq i} F_{\alpha_j}
\end{align*}
and note that $R\succeq 0$, and therefore $R^{1/2}$ is
well-defined, which allows us to define the single-qubit operator
\begin{align}
\label{eq:positive-map}
  Y^{(i)}_\ualpha \EqDef \Tr_{\setminus \{i\}}
    (R^{1/2} \cdot X \cdot R^{1/2}) =
    \Tr_{\setminus \{i\}} (X \cdot R)
    = \Tr_{\setminus \{i\}} \left( X \cdot \bigotimes_{j\neq i} 
        F_{\alpha_j}\right).
\end{align}
Above, $\Tr_{\setminus i}$ denote partial trace on all system but
site $i$. Then on one hand, $Y^{(i)}_\ualpha$ is a PSD operator
because it was obtained from $X$ by a positive
map~\eqref{eq:positive-map}. On the other
hand, by the orthogonality of the dual basis $\{D_\ualpha\}$ with
the SIC-POVM $\{E_\ualpha\}$,
\begin{align*}
  Y^{(i)}_\ualpha = \sum_{\ubeta\in S^{(i)}_\ualpha} 
    w_\ubeta D_{\beta_i}.
\end{align*}
We claim that if the above sum contains less than $3$ elements,
it is not a PSD. To see this, assume without loss of
generality that the sum is given by $A(\lambda) = \lambda D_0 +
(1-\lambda)D_1$ for some $\lambda\in [0,1]$. By the symmetry of
the SIC-POVM basis, our argument applies to any other pair of dual
basis matrices. Now $A(0) = D_1$ and $A(1) = D_0$. At both ends we
get a matrix with a $(-1)$ eigenvalue and therefore at these
points $A(\lambda)$ is not a PSD. If at some $\lambda\in (0,1)$ it
was a PSD, then there should be $\lambda^*\in (0,1)$ at which the
negative eigenvalue of $A$ turns into $0$, and therefore $\det
A(\lambda^*)=0$. However, a direct calculation shows that $\det
A(\lambda) = 4(\lambda^2 - \lambda -\frac{1}{2})$, and this
function does not have a root inside $(0,1)$. 

Using this claim, it is now easy to see that if $X$ is a PSD then
its support must contain at least $3^n$ strings. Indeed, starting
from some string $\ualpha=(\alpha_1, \ldots, \alpha_n)\in
\supp(X)$, there exists at least two other strings in the support
with $\alpha_1$ replaced by $\alpha_1', \alpha_1''$. Each one of
these 3 strings then gives birth to 3 other strings by looking at
$\alpha_2$, etc. Eventually, we end up with at least $3^n$ strings
in the support.

\section{Proof of \Lem{lem:4decomp}}
\label{app:4decomp-proof}

  First we consider a super-graph $G''$ of $G'$ where every vertex
  has degree 3 or 4. Such a super-graph can be constructed by adding
  to $G'$ an at most linear number of 4-cliques on new vertices.
  Each of these new vertices has degree 3, and therefore can be
  connected to a vertex of $G'$ of degree 1 or 2.  By doing this
  sufficiently many times all vertices in $G''$ will have degree 3
  or 4.
  
  We now recall two useful results for the proof. First,
  a result of Thomasen\cc{ref:Thomassen1981-factor}, extending a
  result of Tutte\cc{ref:Tutte1978-subgraph}, states that for every
  $r \geq 1$, every graph where all vertices have degree $r$ or
  $r+1$, has a non-trivial path-cover. Second,
  in\cc{ref:Gomez2020-graph-cover} Gom\'ez and Wakabayashi give an
  efficient algorithm that finds a non-trivial path-cover in a graph
  if there is one.

  By the above, we can efficiently find a non-trivial path-cover for
  $G''$, and the edges of this path-cover that belong to $G'$
  constitute our first output graph $G'_1.$ In the remaining
  subgraph of $G'$ every vertex has degree 1 or 2 or 3. By a process
  analogous to the construction of $G''$, we can construct a
  super-graph $G'''$ of this subgraph where every vertex has degree
  2 or 3. Therefore we can again find efficiently a non-trivial
  path-cover of $G'''$, and the edges of this path-cover that belong
  to $G'$ constitute our second output graph $G'_2$.  The remaining
  subgraph of $G'$ has vertices of degree 1 or 2 and therefore it
  consists of simple cycles, paths and isolated vertices. This is
  our third output graph $G'_3$.

\section{Constructing \texorpdfstring{$\lambda$}{lambda}-solutions 
  for cycles and graphs of the third kind} 
\label{sec:lambda-min}

In this section we prove
Lemmas~\ref{lem:third-kind},~\ref{lem:cycle} by explicitly
constructing $\lambda$-solutions, with $\lambda > 0$, for cycles and
graphs of the third kind.

\subsection{Proof of \Lem{lem:cycle}}
\label{sec:cycle-proof} 

Let $G=(V,E)$ be a cycle with $n=|V|$ vertices, which we denote by
by $i=1,2,\ldots, n$. By Corollary~\ref{corol:lambda-min}, to
show that there exists a $\lambda$-solution to $G$, we need to show
that there exists a color-invariant distribution of legal
$4$-coloring $\mu$ in which the collision probability
$\Prob_\mu(\alpha_i=\alpha_j)$ for every $\{i,j\}\notin E$ is in the
region $[a,b]$ for some global $0<a<b<\frac{1}{3}$. We shall
construct this distribution explicitly for the cases $n=3,4,5$, and
then generally for every $n\ge 6$. 

We start with the edge cases $n=3,4,5$, which are proved in
the following claim.
\begin{claim}
  There exist $\lambda$-solutions for cycles of length $3,4,5$.
\end{claim}

\begin{proof}
  We use the following constructions.
  \begin{itemize}
    \item For $n=3$ all vertices are neighbors of each other so we can 
      simply take $X_3\EqDef X_\sym(\ualpha^*)$ for $\ualpha^*
      \EqDef (0,1,2)$.
      
      \item For $n=4$, we take $\ualpha^*\EqDef (0,1,2,3)$,
      $\ubeta^* \EqDef (0,1,0,1)$ and define $X_4 \EqDef
      \frac{3}{4}X_\sym(\ualpha^*) + \frac{1}{4}X_\sym(\ubeta^*)$.
      The collision probability in $X_\sym(\ualpha^*)$ for every
      non-neighboring vertices is $0$, while in $X_\sym(\ubeta^*)$
      it is $1$. Therefore, in $X_4$ it is exactly $\frac{1}{4}$.
      
    \item For $n=5$: consider a cycle of length $5$ with vertices
    numbered by $1,2,3,4,5$. There are exactly $5$ pairs of
    non-neighboring vertices: $(1,3), (1,4), (2,4), (2,5), (3,5)$.
    For each of these pairs we define a specific solution in which
    we color this pair in the same color (say, 0), and color the
    rest of the vertices with the $3$ remaining color and 
    symmetrize. The collision probability for non-neighboring
    vertices is $1$ for the special pair, and $0$ for the rest.
    Taking the convex combination of all these $5$
    solutions, we obtain a state in which the collision probability
    is $1/5$.
  \end{itemize}
  
\end{proof}

Let us now assume that $G=(V,E)$ is a cycle with $n\ge 6$ vertices,
which we index by $1,2,3, \ldots$. We partition the vertices into
$V_{\rm even}$, $V_{\rm odd}$. When $n$ is even then $|V_{\rm
even}|=|V_{\rm odd}|=n/2$, and when $n$ is odd, then $V_{\rm
odd}=(n+1)/2$ while $V_{\rm even} = (n-1)/2$. These two cases are
depicted in \Fig{fig:loops}. Note that in the even case, no vertices
of the same type are neighbors, whereas in the odd case, the odd
vertices $\{1,n\}$ are neighbors.
\begin{figure}
  \begin{center}
    \includegraphics[scale=0.2]{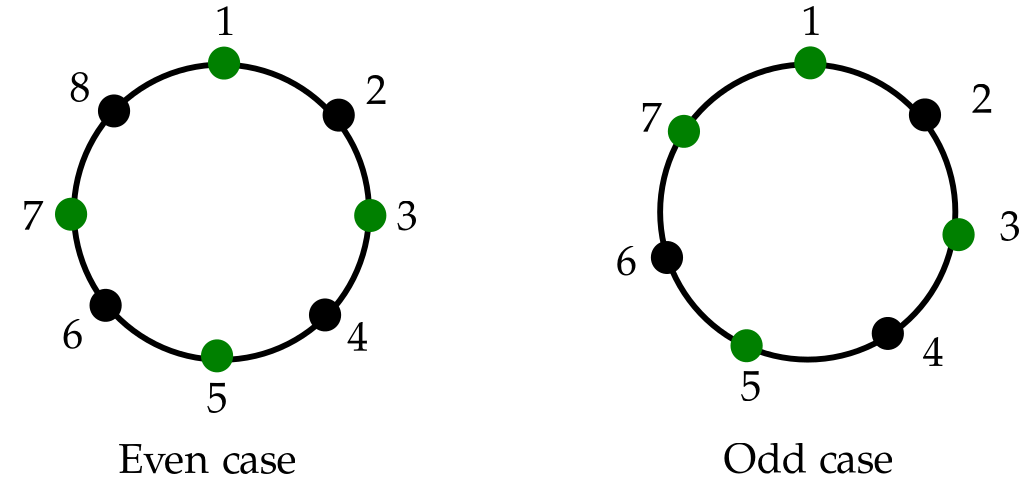}
  \end{center}
  \caption{The construction of a $\lambda$-solution for cycles. We
  index the vertices of the cycle by $\{1,2,\ldots, n\}$ and mark
  the even and odd vertices differently. When the cycle has an odd
  number of vertices, the two odd vertices $1,n$ are adjacent.}
\label{fig:loops}
\end{figure}

We will construct the coloring distribution $\mu$ by first
constructing an auxiliary distribution $\mu'$ that is defined by the
following stochastic procedure. We first color randomly and
uniformly all the even vertices, which by construction are never
adjacent to each other (the black vertices in \Fig{fig:loops}).
Next, we uniformly in random color the odd vertices so that they are
compatible with the coloring of their neighbors. We claim:
\begin{claim}
  There exist $0<a'<b'<\frac{1}{3}$ such that for every $\{i,j\}\notin
  E$, 
  \begin{align*}
    \Prob_{\mu'}(\alpha_i=\alpha_j) \in [a',b']
  \end{align*}
  except for when $i,j$ are odd vertices
  that are next-to-nearest neighbors, in which case
  $\Prob_{\mu'}(\alpha_i=\alpha_j)=\frac{1}{3}$.
\end{claim}
\begin{proof}
  Let us now analyze the collision probability in $\mu'$. For that,
  we use the following notation. With any vertex $i\in V_{\rm odd}$, we
  associate two ``parent vertices'' in $j,k \in V_{\rm even}$, on which
  the coloring of $i$ depends. We will also treat any $j\in
  V_{\rm even}$ as its own parent. Note that since $n\ge 6$, then in the
  even case, two odd vertices can have at most one common parent,
  while in the odd case only the $1,n$ vertices have two common
  parents. As colors in $V_{\rm even}$ are chosen independently, it is
  clear that whenever $i,j$ do \emph{not} share a common parent,
  $\Prob_{\mu'}(\alpha_i=\alpha_j) = \frac{1}{4}$. We therefore need
  to consider the case where $\{i,j\}\notin E$ but they share a
  common parent. We shall consider the cases of $n$ even and $n$ odd
  separately.

  When $n$ is even, this can only happen when $i,j\in V_{\rm odd}$
  and are adjacent to $k\in V_{\rm even}$. By color invariance, 
  \begin{align*}
    \Prob_{\mu'}(\alpha_i=\alpha_j)
      = \Prob_{\mu'}(\alpha_i=0|\alpha_j=0, \alpha_k=1) .
  \end{align*}
  But by construction, $\Prob_{\mu'}(\alpha_i=0|\alpha_j=0,
  \alpha_k=1) = \Prob_{\mu'}(\alpha_i=0|\alpha_k=1)$. Using color
  invariance once more, we conclude that
  $\Prob_{\mu'}(\alpha_i=0|\alpha_k=1)=1/3$, from which we conclude
  that $\Prob_{\mu'}(\alpha_i=\alpha_j) = 1/3$.

  When $n$ is odd, we also need to consider the cases where (i)
  $i=n\in V_{\rm odd}$ $j=2\in V_{\rm even}$, (ii) $i=1$, $j=3$ (iii) $i=n$,
  $j=3$ and (see the odd case in \Fig{fig:loops}):
  \begin{description}
    \item [case (i):] In this case the colors of the odd vertices
      $n,1$ are determined by the colors of the parents $2,n-1$.  If
      $\alpha_2=\alpha_{n-1}$, then $\Prob(\alpha_n=\alpha_2)=0$.
      Therefore,
      \begin{align*}
        \Prob_{\mu'}(\alpha_n=\alpha_2) = \Prob_{\mu'}(\alpha_n=\alpha_2 |
        \alpha_{n-1}\neq \alpha_2)
          \cdot \Prob_{\mu'}(\alpha_{n-1}\neq \alpha_2).
      \end{align*}
      Invoking color invariance, we have:
      \begin{align*}
        \Prob_{\mu'}(\alpha_n=\alpha_2 |\alpha_{n-1}\neq \alpha_2)
         = \Prob_{\mu'}(\alpha_n=0 |\alpha_{n-1}=1, \alpha_2=0).
      \end{align*}
      When $\alpha_{n-1}=1$ and $\alpha_2=0$, then there are only 7 
      possible colorings of the (odd) vertices
      $(\alpha_n,\alpha_1)$, which are: $(0,1), (0,2), (0,3),
      (2,1),(2,3), (3,1),(3,2)$. Three of these colorings have
      $\alpha_n=0$, and so we deduce $\Prob_{\mu'}(\alpha_n=0
      |\alpha_{n-1}=1, \alpha_2=0)=\frac{3}{7}$. Finally, as
      $n-1$,$2$ are even vertices, $\Prob_{\mu'}(\alpha_{n-1}\neq
      \alpha_2)=\frac{12}{16} = \frac{3}{4}$. All together, we get
      \begin{align*}
        \Prob_{\mu'}(\alpha_n=\alpha_2) = \frac{3}{7}\cdot\frac{3}{4}
         = \frac{9}{28} \simeq 0.321.
      \end{align*}

     \item [case (ii):] Here the analysis is identical to the case
       when the cycle is even. Invoking color invariance, 
       \begin{align*}
         \Prob_{\mu'}(\alpha_1=\alpha_3)
           = \Prob_{\mu'}(\alpha_1=0|\alpha_3=0, \alpha_2=1).
       \end{align*}
       By construction, $\Prob_{\mu'}(\alpha_1=0|\alpha_3=0,
       \alpha_2=1) = \Prob_{\mu'}(\alpha_1=0|\alpha_2=1)$.
       Invoking color invariance once more, we conclude that
       $\Prob_{\mu'}(\alpha_1=0|\alpha_2=1)=1/3$ and therefore,
       $\Prob_{\mu'}(\alpha_1=\alpha_3) = 
       \Prob_{\mu'}(\alpha_i=\alpha_j) =
       \Prob_{\mu'}(\alpha_1=0|\alpha_2=1) = \frac{1}{3}$.
   
     \item [case (iii):] Here we can use the result of case (i). Let
       $k=2$ be the parent of both $i=n$ and $j=3$. Then
       \begin{align*}
         \Prob_{\mu'}(\alpha_n=\alpha_3) 
           &= \Prob_{\mu'}(\alpha_n=0|\alpha_3=0, \alpha_2=1)
           = \Prob_{\mu'}(\alpha_n=0|\alpha_2=1)\\
           &= \frac{\Prob_{\mu'}(\alpha_n=0,\alpha_2=1)}
             {\Prob_{\mu'}(\alpha_2=1)},
       \end{align*}
       where the first equality follows from color invariance, the
       second from construction, and the third from definition of
       conditional probability. Next, using color invariance once
       again, we conclude that $\Prob_{\mu'}(\alpha_2=1) = 1/4$ and
       \begin{align*}
         \Prob_{\mu'}(\alpha_n=0,\alpha_2=1) 
           = \frac{1}{12}\Prob_{\mu'}(\alpha_n\neq\alpha_2)
           = \frac{1}{12}\big(1-\Prob_{\mu'}(\alpha_n=\alpha_2)\big).
       \end{align*}
       Using the result of case (i),
       $\Prob_{\mu'}(\alpha_n=\alpha_2)=\frac{9}{28}$ and therefore
       \begin{align*}
         \Prob_{\mu'}(\alpha_n=\alpha_3)
         = 4\cdot\frac{1}{12}\cdot(1-\frac{9}{28}) =
         \frac{19}{84}\simeq 0.226 .
       \end{align*}
  \end{description}
    
  This concludes the proof of the claim.
\end{proof}

To finish the proof of \Lem{lem:cycle}, we need to take care of the
cases where the collision probability is $\frac{1}{3}$. Here the
solution is simple: we define $\mu$ to be the uniform convex
combination of all distributions $\mu'$ that we obtain by shifting
the vertices on the cycle by $2,3,\ldots, n-1$ steps. Necessarily,
the problematic $i,j$ where the original collision probability was
$\frac{1}{3}$ will be averaged with shifts in which it is smaller
than $\frac{1}{3}$, and consequently all the collision probabilities
are in some range $[a,b]$ for $0<a<b<\frac{1}{3}$.

\subsection{Proof of \Lem{lem:third-kind}}
\label{sec:third-kind-proof}

Let $G=(V,E)$ be a graph of the third kind (\Def{def:third-kind}).
We shall build our $\lambda$-solution $X=\sum_\ualpha \mu(\ualpha)
D_\ualpha$ by defining the underlying probability distribution $\mu$
of 4-colorings over the vertices as a convex combination of two
such distributions $\mu_{I}$ and $\mu_{II}$.

By definition, we can write the vertices in $V$ as $V=V_J\cup V_{\rm
mid} \cup V_{\rm op}$, where $V_J$ are the original vertices of $J$,
$V_{\rm mid}$ are the mid vertices added to $V'$ and $V_{\rm op}$ 
are the optional vertices added to $V'''$ that are connected to the
mid vertices. We first define the distribution $\mu_I$ iteratively
as the following stochastic process:
\begin{enumerate}
  \item The colors of the vertices in $V_J$ are picked 
    independently and uniformly randomly from $\{0,1,2,3\}$.
  
  \item We uniformly pick a coloring of the mid vertices $V_{\rm mid}$ 
    that is consistent with the coloring of its $V_J$ neighbors
    
  \item We uniformly pick a coloring of $V_{\rm op}$ that is consistent
    with the coloring of its $V_{\rm mid}$ neighbor.
\end{enumerate}
By Corollary~\ref{corol:lambda-min}, the positivity of the marginals
of any distribution $\mu'$ depends on the collision probability
$\tau_{\mu'}^{(ij)}= \Prob_{\mu'}(\alpha_i=\alpha_j)$, which has to
be in $(0,1/3)$ for any $\{i,j\}\notin E$. Let us calculate it for
$\mu_I$. By construction, the colors of the $V_J$ vertices are
independent, while the color of a mid vertex $i\in V_{\rm mid}$ depends
only on the coloring of its two neighbors in $V_J$. Similarly, the
color of $i\in V_{\rm op}$ depends on the color of its neighbor in $j\in
V_{\rm mid}$, which, in turn, depends only the colors the two neighbors
of $j$ in $V_J$. We can therefore associate with each vertex in
$V_{\rm mid}\cup V_{\rm op}$ two ``parent'' vertices in $V_J$ on which they
depend, and define every vertex in $V_J$ be its own parent. With the
above notation, it is clear that the colors of vertices with no
common parent are completely independent. For such vertices $i,j$,
it is therefore clear that $\Prob_{\mu_I}(\alpha_i=\alpha_j) =
\frac{1}{4}$. Let us analyze the remaining cases.

\begin{figure}
  \begin{center}
    \includegraphics[scale=0.2]{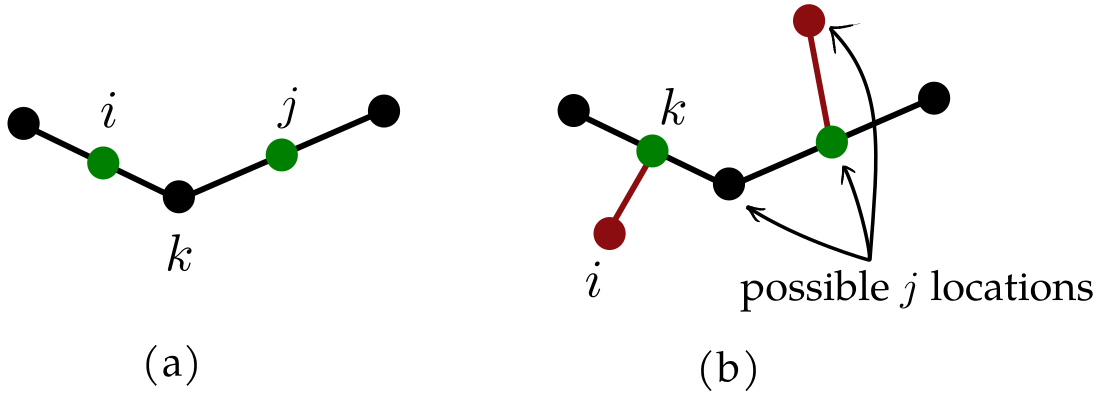}
  \end{center}
  \caption{The special cases of vertices coloring treated in the
  proof of \Lem{lem:third-kind}. (a) $i,j\in V_{\rm mid}$ share a common
  parent $k \in V_J$. (b) $i\in V_{\rm op}$ and $j\in V$ share a common
  parent.}
\label{fig:third-kind-vertices}
\end{figure}

\begin{itemize}
  \item If $i,j\in V_{\rm mid}$ and share a common parent
    $k\in V_J$ (\Fig{fig:third-kind-vertices}(a)).  In such case,
    the color invariance of $\mu_I$ implies that
    $\Prob_{\mu_I}(\alpha_i=\alpha_j) = \Prob_{\mu_I}(\alpha_i=0|
    \alpha_j=0)$. But if $\alpha_j=0$, then clearly $\alpha_k\neq
    0$. Assume without loss of generality that $\alpha_k=1$. In such
    case, $\alpha_i$ is picked uniformly by one of the remaining $3$
    colors which is different its other parent. But since the other
    parent color is picked independently, we conclude that
    $\Prob_{\mu_I}(\alpha_i=0|\alpha_j=0)=\frac{1}{3}$.
    
  \item If $i\in V_{\rm op}$ and $j\in V$ share a parent and
  $\{i,j\}\notin E$ (\Fig{fig:third-kind-vertices}(b)). In such case,
  let $k\in V_{\rm mid}$ be the neighbor of $i$. Then
  \begin{align*}
    \Prob_{\mu_I}(\alpha_i=\alpha_j) 
      &= \Prob_{\mu_I}(\alpha_i=\alpha_j|\alpha_j=\alpha_k)
        \cdot\Prob_{\mu_I}(\alpha_j=\alpha_k)\\
      &+ \Prob_{\mu_I}(\alpha_i=\alpha_j|\alpha_j\neq\alpha_k)
        \cdot\Prob_{\mu_I}(\alpha_j\neq\alpha_k)
  \end{align*}
  By construction, since $i,k$ are neighbors, and $i$ is colored
  uniformly by any color that is different than that of $k$, we have
  $\Prob_{\mu_I}(\alpha_i=\alpha_j|\alpha_j=\alpha_k) =0$ and
  $\Prob_{\mu_I}(\alpha_i=\alpha_j|\alpha_j\neq \alpha_k)
  =\frac{1}{3}$. Therefore,
  \begin{align*}
    \Prob_{\mu_I}(\alpha_i=\alpha_j)
      &= \Prob_{\mu_I}(\alpha_i=\alpha_j|\alpha_j\neq\alpha_k)
        \cdot\big(1-\Prob_{\mu_I}(\alpha_j = \alpha_k)\big)\\
      &= \frac{1}{3}\cdot\big(1
        - \Prob_{\mu_I}(\alpha_j =\alpha_k)\big) .
  \end{align*}
  The collision probability $\Prob_{\mu_I}(\alpha_i=\alpha_j)$ is
  therefore given as a function of the collision probability 
  $\Prob_{\mu_I}(\alpha_j =\alpha_k)$, where now $k\in V_{\rm mid}$ and
  $j,k$ share a common parent. We now consider the $3$ cases,
  illustrated in \Fig{fig:third-kind-vertices}(b). 
  \begin{enumerate}
    \item When $j\in V_J$, then $(j,k)\in E$ and so 
      $\Prob_{\mu_I}(\alpha_j =\alpha_k)=0$, which yields
      $\Prob_{\mu_I}(\alpha_i=\alpha_j)=\frac{1}{3}$.
    
    \item When $j\in V_{\rm mid}$, then by the previous calculation
      $\Prob_{\mu_I}(\alpha_j =\alpha_k)=\frac{1}{3}$, which yields
      $\Prob_{\mu_I}(\alpha_i=\alpha_j)=\frac{1}{3}\cdot\frac{2}{3}
      = \frac{2}{9}$.
      
    \item When $j\in V_{\rm op}$, then by the above result,
      $\Prob_{\mu_I}(\alpha_j =\alpha_k)=\frac{2}{9}$, hence
      $\Prob_{\mu_I}(\alpha_i=\alpha_j)=\frac{1}{3}\cdot\frac{7}{9}
      = \frac{7}{27}$.
  \end{enumerate}
\end{itemize}

To conclude, for all $\{i,j\}\notin E$, the collision probability is
either $\frac{1}{4}, \frac{2}{9}$, or $\frac{7}{27}$, except for the
cases (a) when $i,j\in V_{\rm mid}$ and share a common parent, and (b)
$i\in V_{\rm op}$ and $j\in V_J$ is its parent. In these cases, the
collision probability is $\frac{1}{3}$. 

Recall that in order to show strict positivity, the collision
probability must be bounded above from $0$ and below from
$\frac{1}{3}$. We therefore need to ``fix'' the two $\frac{1}{3}$
cases above. We do that by constructing a second coloring
distribution $\mu_{II}$ in which the collision probability for these
two cases is bounded away from 1/3 from below. Then taking the
convex combination $\mu = (1-\eps)\mu_I + \eps\mu_{II}$ for
sufficiently small (but constant) $\eps$ will prove the lemma. 

To motivate the definition of $\mu_{II}$, let us concentrate on the
first problematic case (case (a)). The simplest solution would be to
find a 4-coloring of the mid vertices such that for every pair of
vertices $i,j\in V_{\rm mid}$ with a common $V_J$ neighbor have
different colors. This is the same as asking for a legal 4-coloring
of the edges of the graph $J$ (meaning that adjacent edges have
different colors). Such an edge coloring does not necessarily
exists, but Vizing's theorem\cc{ref:Vizing1965} asserts that there
exists a legal edge coloring with 5 colors.  In order not to confuse
the coloring of the edges with the coloring of the vertices, we
shall refer to these $5$ colors as ``labels''. We then could map
randomly the 5 edge labels to the 4 vertex colors. In fact, it is
more useful to map the labels into only 3 vertex colors, this would
enable us to deal easily with the second problematic case. Though
the collision probability will be positive, it will still be bounded
away from below from 1/3, which is what we want.

Formally, we construct $\mu_{II}$ as follows. We fix a legal
5-coloring of the edges. We pick a uniformly random mapping $m :
\{a,b,c,d,e\} \rightarrow \{0,1,2\}$ which satisfies $1 \leq
|m^{-1}(y)| \leq 2$, for $0\leq y \leq2$, where $|m^{-1}(y)|$ is the
number of labels that are mapped to the color $y$. Then by
definition, the color of a mid vertex $i$ is $m(\ell)$, where $\ell$
is the label of the parent edge of $i$ in the fixed edge coloring.
We then color the $V_J$ vertices with the color 3.  Then for each
vertex in $V_{\rm op}$ we choose uniformly randomly between the two
colors in $\{0,1,2\}$ that are different from the color of its
neighbor in $V_{\rm mid}$. Finally, we symmetrize the resultant
distribution by summing over all possible $4!=24$ color permutations
(see \Sec{sec:4-coloring}). 

Let us analyze the collision probabilities in $\mu_{II}$ for the
problematic cases a,b above. In case b, we have $i\in V_{\rm op}$ and
$j\in V_J$. By construction, the $V_J$ vertices are colored
differently that the $V_{\rm mid}, V_{\rm op}$ vertices, and so here the
collision probability is $0$. In case a, $i,j\in V_{\rm mid}$ are
associated with adjacent edges in $J$. In such case, the collision
probability is the probability that the two different labels of the
edges were mapped by $m$ to the same color, which is given by
\begin{align*}
  \Prob_{\mu_{II}}(\alpha_i=\alpha_j) 
    &= \Prob_{\mu_{II}} (m^{-1}(\alpha_j)
   = m^{-1}(\alpha_i) ~|~ |m^{-1}(\alpha_i)| = 2)
     \cdot \Prob_{\mu_{II}}(|m^{-1}(\alpha_i)| = 2) \\
   &= \frac{4}{5}\cdot \frac{1}{4} = \frac{1}{5} .
\end{align*}

Having established that in all problematic cases,
$\Prob_{\mu_{II}}(\alpha_i=\alpha_j) \le \frac{1}{5}$, we use
$\eps=\frac{1}{20}$ and define
\begin{align}
  \mu \EqDef (1-\eps)\mu_I + \eps\mu_{II}.
\end{align}
It can then be verified explicitly that whenever $\{i,j\}\notin E$,
then $\Prob_\mu(\alpha_i=\alpha_j)\le 0.327 < \frac{1}{3}$.

\bibliographystyle{alpha}
\bibliography{REFS}

\end{document}